\documentclass[11pt]{article}
\usepackage{fullpage}
\usepackage{ifthen}
\usepackage{xspace}
\usepackage{amsmath, amssymb, amsfonts}
\usepackage{amsthm}
\usepackage{bm}
\usepackage[pdftex]{graphicx}
\usepackage[english]{babel}
\usepackage[showdeletions]{color-edits}
\usepackage{algorithm}
\usepackage{algorithmic}

\begin{document}

\newcommand{\tr}{\textsf{triangle}}
\newcommand{\li}{\textsf{line}}
\newcommand{\xx}[3]{{#1}^{#2}_{#3}}
\newcommand{\mc}[1]{\mathcal{#1}}
\newcommand{\BRi}[1]{\mathcal{B}^{R_i}(#1)}
\newcommand{\bK}{\mc{\bar K}}
\newcommand{\bq}{\bar q}
\newcommand{\Ex}[1]{\mathop{\mathrm{E}}\left[{#1}\right]}
\newcommand{\ex}{\mathop{\mathrm{E}}}
\newcommand{\pr}[1]{\Pr\left[{#1}\right]}

\newtheorem{theorem}{Theorem}
\newtheorem{claim}{Claim}
\newtheorem{corollary}[theorem]{Corollary}
\newtheorem{proposition}[theorem]{Proposition}
\newtheorem{lemma}[theorem]{Lemma}
\newtheorem{property}{Property}
\newtheorem{fact}[theorem]{Fact}

\newtheorem{definition}{Definition}
\newtheorem{example}{Example}
\newtheorem{assumption}[theorem]{Assumption}

\newtheorem{remark}{Remark}

\addauthor{pc}{red}

\title{Generalized Mirror Descents in Congestion Games\thanks{%
Part of the results in this paper have appeared in preliminary form
in the proceedings of AAMAS 2014 \cite{CL14} and as an extended abstract in AAMAS 2015 \cite{CL15}.}}

\author{%
  Po-An Chen\thanks{%
    Institute of Information Management, National Chiao Tung University, Taiwan.
    Email: poanchen@nctu.edu.tw.
    Supported in part by NSC 102-2221-E-009-061-MY2.
  }
  \and
  Chi-Jen Lu\thanks{%
    Institute of Information Science, Academia Sinica, Taiwan.
    Email: cjlu@iis.sinica.edu.tw.
  }
}
\date{}

\begin{titlepage}
\maketitle

\begin{abstract}
Different types of dynamics have been studied in repeated game play, and one of them which has received much attention recently consists of those based on ``no-regret" algorithms from the area of machine learning. It is known that dynamics based on generic no-regret algorithms may not converge to Nash equilibria in general, but to a larger set of outcomes, namely coarse correlated equilibria. Moreover, convergence results based on generic no-regret algorithms typically use a weaker notion of convergence: the convergence of the average plays instead of the actual plays. Some work has been done showing that when using a specific no-regret algorithm, the well-known multiplicative updates algorithm, convergence of actual plays to equilibria can be shown and better quality of outcomes in terms of the price of anarchy can be reached for atomic congestion games and load balancing games. Are there more cases of natural no-regret dynamics that perform well in suitable classes of games in terms of convergence and quality of outcomes that the dynamics converge to?

We answer this question positively in the bulletin-board model by showing that when employing the \emph{mirror-descent} algorithm, a well-known generic no-regret algorithm, the \emph{actual} plays converge quickly to equilibria in nonatomic congestion games. This gives rise to a family of algorithms, including the multiplicative updates algorithm and the gradient descent algorithm as well as many others. Furthermore, we show that our dynamics achieves good bounds on the outcome quality in terms of the price-of-anarchy type of measures with two different social costs: the average individual cost and the maximum individual cost.

Finally, the bandit model considers a probably more realistic and
prevalent setting with only partial information, in which at each time step each player only knows the cost of her own
currently played strategy, but not any costs of unplayed strategies.
For the class of atomic congestion games, we propose a family of bandit algorithms based on the mirror-descent algorithms
previously presented, and show that when each player individually adopts such a bandit
algorithm, their joint (mixed) strategy profile quickly converges with implications. 
\end{abstract}


\end{titlepage}

\section{Introduction}
Nash equilibrium is a widely-adopted solution concept in game theory, which is used for predicting the outcomes of systems consisting of self-interested players. We are interested in repeated game play, and a Nash equilibrium describes a steady state in which the system would stay once it is reached. However, this raises the issue of how such a state can be reached. In fact, for a general game, computing a Nash equilibrium is believed to be hard (according to the PPAD-hardness results \cite{chen:deng:settling}), so an equilibrium may not be reached in a reasonable amount of time in general, and the outcomes that we have observed may all be far out of any equilibrium, which would render the study on equilibria meaningless.
To address this issue, a line of research is to consider natural efficient dynamics which players have incentive to follow, and study how the system evolves according to such dynamics.
\begin{itemize}
\item \textbf{Best or better response dynamics.} One natural dynamics is the best or better response dynamics, in which a deviating player at each time makes a best or better change in his/her strategy to improve his/her payoff given the current choice of the other players. This means that, for a player to deviate, there must be enough information regarding the current choice that the others had made.
    It is well-known that such dynamics leads to pure Nash equilibria in congestion games. However, a player may not have incentive to play this way because making such deviations may not be beneficial if other players also deviate at the same time.

\item \textbf{Generic no-regret dynamics.} One may argue that a plausible incentive for a player is to maximize his/her average payoff through time, and dynamics based on ``no-regret" algorithms from the area of online learning (e.g., \cite[Chapter 4]{NRTV07}) have thus been proposed in the study.
    The no-regret property is preserved in full or partial information models of feedback by a variety of algorithms.
    For a nonatomic routing game, it is known that if each infinitesimal player plays any arbitrary no-regret algorithm, the ``time-averaged" flow and flows at \emph{most} time steps would be at some type of approximate Nash equilibrium \cite{blum:even-dar:ligett:routing}.
    For a ``socially concave" game, a similar time-averaged convergence result is also known \cite{even-dar:mansour:nadav}.\footnote{Note that the games that we consider here are not socially concave.}
Convergence to a Nash or approximate Nash equilibrium is not always the case in general, and playing arbitrary no-regret algorithms can result in a larger set of outcomes than Nash equilibria, namely \emph{coarse correlated equilibria}.\footnote{See, for example, \cite[Proposition~3.1]{roughgarden:lecture} among others that discuss this.} Nevertheless, if one only cares about the outcome quality and the quality is measured by the price of anarchy \cite{koutsoupias:papadimitriou:anarchy} with the average individual cost, it is known that
the price of total anarchy achieved by such no-regret algorithms can still match the price of anarchy 
at Nash equilibrium in special games, such as atomic congestion games \cite{blum:hajiaghayi} or even a wider class of smooth games \cite{roughgarden}.
On the other hand, there are broad classes of games and natural measures of outcome quality for which large gaps are known between no-regret outcomes and Nash equilibria.

Notice that the convergence results mentioned above are, instead of the convergence of the actual strategy, about the convergence of the time-averaged strategy \cite{blum:even-dar:ligett:routing,even-dar:mansour:nadav} or flows at most time steps being close to equilibria \cite{blum:even-dar:ligett:routing}.
Even in the latter case, those time steps where flows are close to equilibria are arbitrarily distributed over time (not guaranteed to gather toward the end in time),
which means that a flow at some very late point in time can still be far away from any equilibria and flows may not stabilize.
The guarantees are still \emph{not} on the convergence of the actual plays.
Such results are useful if the goal is to solve the computational problem of computing an approximate Nash equilibrium,
but they may not tell us much about how the system actually evolves.
In particular, even though the time-averaged play converges to an equilibrium or most plays are close to equilibria,
the actual strategy may \emph{not} converge and may be far away from an equilibrium.
For many applications, for example, that require the system to stabilize, the time-averaged play convergence or most plays being close to equilibria may not be enough.



\item \textbf{Multiplicative updates dynamics with full information.} Although it is nice to be able to have general positive results on what generic no-regret algorithms can achieve, one may wonder if going from generic no-regret algorithms to specific ones could yield stronger results, in terms of convergence or quality of outcomes that the algorithms converge to. One of the best known no-regret algorithms is the \emph{Multiplicative Updates} (MU) algorithm \cite{littlestone:warmuth,freund:schapire}. Kleinberg et al. \cite{kleinberg:piliouras:tardos:multiplicative} studied this for atomic congestion games in the full information setting, in which players have full information about the cost functions so that they can determine the cost of every other strategy they could have used given other players strategies at the current round. It was shown that if each player employs such an MU algorithm, the actual joint mixed strategy profile of players converges to a pure Nash equilibrium with high probability for most games. Note that here it is the actual joint strategy profile, instead of the time-averaged one, which converges. Furthermore, since the set of pure Nash equilibria can be a very small subset of correlated equilibria, the price of total anarchy achieved this way can be much smaller than that by a generic no-regret algorithm.

\item \textbf{Multiplicative updates dynamics with bulletin-board posting.} In another work \cite{kleinberg:piliouras:tardos:load}, Kleinberg et al. studied the smaller class of load balancing games, but in the more stringent partial-information setting of the ``bulletin-board" model, in which players only know the actual cost value of each edge according to the actual strategies played at the current round. They showed that if all the players play according to a common distribution (i.e., mixed strategy) and update the distribution using such an MU algorithm, the common distribution converges to some symmetric equilibrium of the \emph{nonatomic} version of the game. As a result, the price of total anarchy achieved this way is also considerably smaller than that by a generic no-regret one.
However, their analysis relies crucially on the assumption that all the players at each round play according to the same distribution. This assumption may not be reasonable in other settings or in other games, which makes the applicability of their analysis somewhat limited. On the other hand, the analysis in \cite{kleinberg:piliouras:tardos:multiplicative} can do without the assumption and deal with general asymmetry in players' probability distributions, but it only works in the full information model.
\end{itemize}

Note that there is no equilibrium selection for the best (or better) response dynamics and generic no-regret dynamics,
and they could converge to the worst corresponding equilibrium.
Nonetheless, the results of multiplicative updates dynamics suggest that the dynamics converge to a subset of mixed outcomes, namely,
pure Nash equilibria with high probability in \cite{kleinberg:piliouras:tardos:multiplicative} and quite uniformly distributed mixed Nash equilibria in the case of bulletin-board load balancing \cite[Lemma~6]{kleinberg:piliouras:tardos:load};
the price-of-anarchy type of efficiency gets better since the worst (mixed) Nash equilibrium in the induced subset could be better than the worst the worst coarse correlated equilibrium, meaning equilibrium is selected by the dynamics.

These results of multiplicative updates, which form a good comparison and complement to each other, along with the results on generic no-regret plays motivate our quest for other classes of learning dynamics in suitable classes of games and settings.
\emph{Are there more cases of natural no-regret dynamics that perform well in suitable classes of games in terms of convergence time and quality of outcomes that the dynamics converge to?}
We first answer this question positively by providing a family of such dynamics in the \emph{bulletin-board} model for the class of nonatomic congestion games with cost functions of bounded slopes. More precisely, we show that in such a game, if each infinitesimal player individually plays some type of the \emph{mirror-descent} algorithm \cite{BT03}, a well-known general no-regret algorithm, then their joint strategy profile quickly converges to an approximate notion of Wardrop equilibrium.\footnote{Wardrop equilibrium can be seen as Nash equilibrium specialized for games with infinitely many agents \cite{wardrop} such as nonatomic congestion games.}
We also show that our dynamics achieves good bounds on the quality of outcomes in terms of the price-of-anarchy type of measures with two different social costs: the average individual cost and the maximum individual cost.

All the previously mentioned results are based on somewhat generous information models. For instance of congestion games, edge cost functions are assumed common knowledge in the full information model of \cite{kleinberg:piliouras:tardos:multiplicative} so that players can determine the costs of currently unplayed strategies if they were used. A bit more stringent information model than full information was considered for the load-balancing games in \cite{kleinberg:piliouras:tardos:load} and for the general congestion games in our results previously mentioned,
in which the edge cost functions are not common knowledge anymore, but still the cost values of all paths at each step are assumed available through ``bulletin-board" posting. With such global information, players can get a grasp of the costs corresponding to played and even unplayed strategies, which allows players to update their strategies better and makes convergence of the whole system potentially easier. However, such an assumption on the information availability may not always be realistic and may limit the applicability of these results.

The ``bandit" model in online learning on the other hand considers a probably more realistic and prevalent setting, in which at each time step each player only knows the cost of her (or his) own currently played strategy, but not any costs of unplayed strategies. This gives rise to the dilemma between exploration and exploitation which players have to face. In the area of online learning, many bandit algorithms with no-regret guarantee have been developed, including for example those based on the multiplicative updates algorithm for the experts problem \cite{auer:cesa-bianchi} and those based on the gradient-descent algorithm for online linear or convex optimization \cite{abernethy:haza:rakhlin,flaxman:kalai}. However, not much is known in the area of game theory for playing repeated games in the bandit model. Although similar convergence results of \emph{average} plays can be established immediately for bandit algorithms with no-regret guarantee, we are not aware of any previous result establishing convergence of actual plays in the bandit setting. In fact, it is not clear how to design bandit algorithms with such convergence guarantee.

A natural attempt, following the standard approach for designing bandit algorithms such as those in \cite{auer:cesa-bianchi,abernethy:haza:rakhlin,flaxman:kalai}, is to come up with estimates of the true cost values and feed these estimates to a full information algorithm, to replace the true cost values that are available in the full information setting. For this approach to work for most problems in online learning, it simply suffices to guarantee these estimates being ``unbiased", in the sense that their \emph{expected} values equal the true cost values.
However, in the setting of repeated game playing, using such unbiased estimates does \emph{not} seem to ensure convergence of actual plays in general. This is because these estimates, even with guarantee on expected values, can still have high variance and thus have actual values very different from the true cost values.
Unfortunately, this is indeed the case for adapting most existing bandit algorithms, as their estimates actually can take very different values from the true cost values with some probability (for example, see the one-point gradient estimate in \cite{flaxman:kalai}) although these estimates are enough for their goal of just achieving no regret therein. In particular, when using such estimates in, for example, mirror descents, each update step may go in a very different (possibly in almost opposite) direction from the desired one according to the true cost values, which does not seem likely to result in convergence to equilibrium. This motivates us to ask the question: are there natural classes of \emph{bandit} algorithms which selfish players individually have incentive to adopt (by the no-regret property) and the whole system will quickly converge to an approximate Nash equilibrium (even just for some classes of instances) with social cost guarantees?

We then answer this more challenging question affirmatively as well. For the class of atomic congestion games, we propose a family of bandit algorithms based on the mirror-descent algorithms presented in the first-half part, and we mainly show that when each player individually adopts such a bandit algorithm, their joint strategy profile quickly converges.
The reasons why we focus on atomic congestion games instead of nonatomic ones are two-folded:
the bandit algorithm for atomic congestion games can be applied for nonatomic congestion games,
by treating the joint (mixed) strategy profile of an atomic game directly as the joint (pure) strategy profile (i.e., flow distribution) of a nonatomic game;
the bandit informational setting is not well defined for nonatomic congestion game since players can split their flows to literally all the allowed paths and get all the path costs.

In the bandit model, each player can only update her own strategy according to very limited and local information about the whole system,
but we show that when each player individually adopts any such bandit algorithm, the whole system still quickly converges, and is to an approximate mixed-strategy equilibrium that has a small approximation error in many natural cases (but has a large approximation error in general) and is beneficial to the society as a whole.
This is not only reminiscent of the result of convergence to some specific mixed Nash equilibria in load balancing with bulletin posting \cite{kleinberg:piliouras:tardos:load} (so a better price-of-anarchy type of efficiency is possible), but also more generally for atomic congestion games like in \cite{kleinberg:piliouras:tardos:multiplicative}.
This may appear even less expected than that in the bulletin-board or full-information model,
where each player at least has more abundant and more global information available.

We provide definitions and some preliminaries in Section~\ref{sec:preliminaries}. First, the generalized mirror-descent algorithm and convergence result in the bulletin-board model are presented in Section~\ref{sec:bulletin}. The bandit algorithm and convergence result are then presented in Section~\ref{sec:bandit}.
Approximate equilibria and the outcome quality bounds in terms of the price-of-anarchy type of measures are discussed along with the convergence results.
We summarize with conclusions and future work in Section~\ref{sec:conclusions}.

\subsection{Discussion of Our Results and Techniques}
\subsubsection*{Bulletin-Board Model}
The mirror-descent algorithm in fact can be seen as a family of algorithms. By instantiating it properly, one can recover the MU algorithm, the gradient-descent algorithm, as well as many others, and our result establishes the fast convergence of all these algorithms at once.
Let us stress that as in \cite{kleinberg:piliouras:tardos:multiplicative,kleinberg:piliouras:tardos:load}, our notion of convergence is the stronger one: what converges is the actual joint strategy profile. 
Note that in the congestion game, different players naturally have different sets of strategies, so it is no longer reasonable to assume that all the players use the same distribution to play as in \cite{kleinberg:piliouras:tardos:load}. Therefore, we allow players to use different distributions and moreover, we allow players to update according to different learning rates. Still, we manage to prove the convergence, just as \cite{kleinberg:piliouras:tardos:multiplicative} but in the more difficult bulletin model and with a concrete bound on convergence time.

Furthermore, we provide bounds on the price-of-anarchy type of measures achieved by our dynamics, in terms of the average individual cost and the maximum individual cost. Using the average individual cost as the social cost, we show that the ratio between the social cost achieved by our dynamics and the optimal one approaches some constant, which depends on the slopes of the cost functions. Using the maximum individual cost as the social cost, we show that the ratio between the social cost achieved by our dynamics and the optimal one also approaches the same constant in symmetric games. In each case, there is a tradeoff between the ratio we can achieve and the time it takes: by letting the system evolve for a longer time, it will get closer to an equilibrium, and the resulting ratio will approach closer to that constant.

Our main technical contribution is the convergence of our dynamics to an approximate equilibrium. To show this, we consider a smooth \emph{convex} potential function of the game which has the joint strategy profile of players as its input. The interesting observation is that although each player individually applies the mirror descent algorithm to his/her own strategy using costs related only to him/her, we show that the updates performed by all the players collectively can be seen as following some generalized mirror descent process on the potential function. The generalized mirror descent allows different step sizes in different dimensions,\footnote{This is similar to adaptive optimization methods such as in \cite{Duchi10adaptivesubgradient}.} and we need this generalization because we allow different learning rates for different players. The standard mirror descent, on the other hand, has the same step size across all the dimensions, so that it moves at each time in exactly the opposite direction of the gradient vector. It is known that doing the standard mirror descent on a smooth convex function leads to a fast convergence to its minimum \cite{BDX11,nesterov}. However, our generalized mirror descent no longer moves in the opposite direction of the gradient vector as different step sizes have different scaling effects in different dimensions, and therefore it is not clear if the process would still converge. Interestingly, we show that a similar convergence result can also be achieved, which may be of independent interest (since this works for all games with a smooth convex potential function satisfying some properties). Finally, let us remark that the standard mirror descent algorithm, instead of the generalized one, has also been used for different problems in game theory: for finding market equilibria in Fisher markets \cite{BDX11} and convex potential markets \cite{CCD13}. Our convergence result for the generalized mirror descent algorithm is an extension of that for the standard one. 

\subsubsection*{Bandit Model}
To be able to achieve convergence in the bandit model, the first hurdle we have to overcome is for each player to have good enough estimates for the true costs of all her allowed paths. As discussed before, we would like these estimates to have actual values (rather than expected values) close to the true cost values, but the estimation methods used in \cite{auer:cesa-bianchi,abernethy:haza:rakhlin,flaxman:kalai} do not work. In fact, an apparent difficulty is that a player can only learn the cost of one single path at each time, but the cost of each path actually depends on how other players choose their paths at that time, which may be very different at different times.
Then how can a player possibly obtain good estimates for the true costs of those unchosen paths at that time? Inspired by the bandit algorithm in \cite[Chapter 4.6]{NRTV07}, we consider dividing the time steps into episodes and letting each player play the same (mixed) strategy at each step during an episode. The expected cost of each path then becomes the same for each step in an episode, and this allows a player to obtain a good estimate for the expected cost of each allowed path, simply by choosing each path an enough number of steps and averaging the costs. With such good estimates, each player can then update her strategy for the next episode by feeding these estimates to the bulletin-board mirror-descent algorithm, to replace the true path costs it needs.

To prove our convergence results, we would like to follow the approach used in the bulletin-board model by showing that the collective update of all the players together corresponds to doing some generalized mirror descent on a convex function. However, as we consider the atomic version of the congestion game instead of the nonatomic version, there are more hurdles that we need to clear.
The first is the choice of the function for analyzing the convergence of our dynamics,
as the potential function $\Phi$ used in the bulletin-board model is defined over nonatomic flows of players.
Actually, the function used to show the convergence of dynamics does not have to be the potential function (for ensuring existence of equilibria) of the game under study as long as the converged outcomes via such used function can be interpreted in the game under study.
We find the same function $\Phi$ still suitable for us, by seeing each player's mixed strategy as a nonatomic flow.
Note that this has been similarly used in atomic load balancing \cite{kleinberg:piliouras:tardos:load} where their mixed strategies are like converging to equilibrium flows of the nonatomic version of load-balancing games, translated into quite uniformly distributed mixed Nash equilibrium in the atomic version.
Yet, this causes a subtle problem. Namely, the gradient of $\Phi$ actually corresponds to path costs according to nonatomic flows rather than those according to atomic flows that players in our atomic game have access to.\footnote{Note that exactly the same problem was faced in studying multiplicative updates in atomic load balancing under the bulletin-board model \cite{kleinberg:piliouras:tardos:load},
where they eventually showed the convergence of mixed strategies using the nonatomic version of load-balancing games and its potential function.
Even without a bandit problem, an error due to this occurred there, too.} This results in a non-negligible amount of error in the estimation of the gradient vector, which seems unavoidable for nonlinear cost functions, and we can at best do an ``approximate" mirror descent with such an approximate gradient vector.

However, the convergence analysis in the bulletin-board model relies crucially on being able to move (in the mirror space) precisely in the oppositive direction of the gradient vector, in order to guarantee that the $\Phi$ value decreases, while it is not hard to find cases with increased $\Phi$ value when moving in a slightly different direction. We bypass this difficulty by showing that as long as the current $\Phi$ value compared to the minimum one is still considerably large, relative to the error of the approximate gradient vector, the next $\Phi$ value will decrease from the current one by some large amount. This provides us a way to bound the number of steps needed to reach a $\Phi$ value within some distance of the minimum one, with the distance dictated by the errors of the approximate gradient vectors.

The convergence also implies an approximate equilibrium in mixed strategies\footnote{There can be alternative definitions of approximate equilibrium.} with a caveat that the approximation can be bad in general to make such equilibrium meaningless.
Nevertheless, there are broad classes of instances when such approximate equilibrium is indeed meaningful:
the approximation error can be small for some classes of cost functions,
for example, linear cost functions and even more generally, when bounds (constant with respect to the amount of flow, not necessarily to the inputs of an instance) on second derivatives are small enough;
for some classes of structures of allowed paths, the approximation error can still be small, for example, in the load-balancing setting like in \cite{kleinberg:piliouras:tardos:load}.

\subsection{Related Work}
There are numerous works regarding reaching various notions of equilibria in congestion/potential games in general: PLS-completeness \cite{fabrikant} and inapproximability \cite{skopalik} of computing Pure Nash equilibria, best-response types of dynamics for converging to approximate equilibria \cite{chien,awerbuch} or connected-component-based ``sink equilibria" \cite{goemans}.
There are studies of no-regret algorithms in zero-sum game play \cite{daskalakis},
game play by selecting strategy profiles to query the corresponding payoffs \cite{fearnley}, etc.
Playing arbitrary no-``swap"-regret algorithms converges to correlated equilibria \cite[Chapter 4.4.3]{NRTV07}
while playing arbitrary no-regret algorithms results in coarse correlated equilibria \cite[Proposition~3.1]{roughgarden:lecture}.
That is to say that the no-regret property is enough to guarantee (coarse) correlated equilibria in general.
On the other hand, our focus is to propose specific classes of no-regret algorithms for players to have incentives to adopt,
with actual convergence guarantee even with bandit feedback.

Our modeling framework in this article is most similar to that in \cite{kleinberg:piliouras:tardos:multiplicative}.
A vector of the probabilities for the all available actions for each player is maintained.
Players sample an action according to this distribution at each time.
Initially, the probabilities can be all equal, i.e., uniform distribution.
Every time each player updates the weights multiplicatively preferring actions of low cost,
which generalizes the weighted majority algorithm introduced by Littlestone and Warmuth \cite{littlestone:warmuth}
and the Hedge algorithm of Freund and Schapire \cite{freund:schapire}.
Kleinberg et al. showed that if players use such dynamics to adjust their strategies in atomic congestion games,
then game play converges to a subset of mixed Nash equilibria, so-called weakly stable equilibria.
Pure Nash equilibria are weakly stable by definition, and the converse was
shown true with probability 1 when congestion costs are selected at random independently
on each edge.

Kleinberg et al. \cite{kleinberg:piliouras:tardos:load} also studied the performance of
learning algorithms in load-balancing games, i.e., congestion games on parallel
links, under the ``bulletin board model" in which players assess edge costs according
to the actual cost incurred on that edge, and not the hypothetical cost if the
player had used it.
This algorithm for specifying a mixed strategy at each time step is a
version of the Hedge algorithm \cite{freund:schapire}, modified so that players assess edge costs
according to the actual cost incurred on that edge, and not the hypothetical
cost if the player had used it for players that do not use the edge at this time step.
It was shown that the bulletin board variant of Hedge is also a no-regret learning algorithm.
Their main result is that the expected makespan of the outcome is bounded
by $O(\log n)$ where $n$ is the number of links/players,
exponentially better than the known lower bounds for arbitrary no-regret algorithms.
Many of our assumptions regarding atomic splittable congestion games
follow \cite{kleinberg:piliouras:tardos:load},
and the analyses for convergence share some high-level intuitions.
 	
Even-Dar et al. studied a subclass of concave games called socially concave games in \cite{even-dar:mansour:nadav}.
They showed that if each player follows any no-external regret minimization procedure,
then the dynamics will converge in the sense that the average action vector will converge to a Nash equilibrium.
Even if we change convexity to concavity and costs to utilities in our paper,
potential games that we consider here are not socially concave games.
Thus, their results do not directly apply.

Besides the dynamics based on no-regret learning algorithms, some other works design \emph{Markovian} rerouting policies in congestion games \cite{fisher,ackermann},
where an agent's behavior only depends on the outcome of the immediately previous round and not on all the previous rounds.
In nonatomic congestion games, inspired by so-called \emph{replicator dynamics} \cite{weibull} the algorithms of Fisher at al. \cite{fisher}, for an agent, adaptively sample paths with a probability proportional to the fraction of agents using this path and reroute with a chosen probability if the latency of the sampled path is smaller than that of the current one. They showed a bicriteria result, an upper bound on the total number of rounds in which it does not hold that almost all agents have a latency close to the average latency. Note that such ``equilibria" are transient in the sense that these equilibria can be left  again even after they are reached. This is more like in \cite{blum:even-dar:ligett:routing}, but \emph{not} the actual convergence in our stronger sense.
Nevertheless, with exploration uniformly at random on unused paths, convergence to a Wardrop equilibrium can now be achieved,
and for symmetric games they gave a polynomial bound on the number of rounds taken to get close to the optimal potential value.

Ackermann et al. \cite{ackermann} used what they called ``concurrent imitation dynamics", which is very similar to the rerouting policies of \cite{fisher}, in atomic congestion games and mainly gave similar results for symmetric atomic congestion games where the analysis needs to take probabilistic effects into account. They first showed a convergence to a local minimum of the potential in pseudopolynomial time.
Their main result is a stronger bound on the expected time to reach an approximate stable state in which at most a very small fraction of the agents deviate by more than a very small fraction from the average latency. Finally, with a suitable combination of imitation with exploration that samples other strategies directly, they guaranteed convergence to Nash equilibria and gave the expected convergence time.

For a generic two-player coordination game, Mehta et al. \cite{mehta} showed that, starting from all but a zero measure of initial probability distributions, a discrete multiplicative weight update algorithm known as \emph{discrete replicator dynamics} converges to \emph{pure} Nash equilibria.
This is not like the randomized results in \cite{kleinberg:piliouras:tardos:multiplicative,kleinberg:piliouras:tardos:load}.
Their results only require that any row/column of the payoff matrix consist of distinct entries, and hold even if the game has uncountably many Nash equilibria.

There is still a variety of different dynamics in repeated games.
Auletta et al. \cite{auletta:ferraioli} presented general bounds on the mixing time of ``logit" dynamics for classes of strategic games.
In the logit dynamics, individual participants act selfishly and keep responding according to some partial noisy knowledge in the complex system.
In particular, they proved nearly tight bounds for potential games and games with dominant strategies.
Kleinberg et al. \cite{kleinberg:ligett} analyzed a game with a unique Nash equilibrium,
but where natural learning dynamics only cycles, not converging to this equilibrium.
They showed that the outcome of this learning process is optimal and has much better social welfare than the unique Nash equilibrium.
Balcan et al. \cite{balcan:blum:mansour} showed that convergence may not lead to any meaningful notions of equilibria,
but may result in good efficiency in terms of some objectives.

\section{Preliminaries} \label{sec:preliminaries}
\paragraph{Nonatomic Congestion Games.}
In this article, we first consider the nonatomic congestion game described by $(N,E,(\mathcal{S}_i)_{i\in N},(c_e)_{e\in E})$,
where $N$ is the set of commodities, $E$ is the set of edges (resources),
$\mathcal{S}_i \subseteq 2^{E}$ is the collection of the allowed paths (the allowed subsets of resources) for commodity~$i$,
and $c_e$ is the cost function of edge~$e$, which is a nondecreasing function of the amount of load on it.
A commodity is a ``pseudo-player", which we simply call a ``player" throughout the discussion of nonatomic congestion games,
meaning that a player herself does not act as a selfish party, but the infinitesimal parties that a player is composed of do.
We describe it in more detail in the following.

Let us assume that $N = \{1,\dots,n\}$, $|E|=m$, and each player has a load of $1/n$ (so the total load is $1$).
Each player consists of a huge infinite number of selfish agents (or see each player as a group of infinitesimal players of the same type).
The agents of player $i$ split the load of player $i$ so that each agent has a small (infinitesimal) amount $\Delta$ of load, i.e., $\Delta\rightarrow 0$.
Each agent of player $i$ must choose one single path $s$ from $\mc{S}_i$ and put that $\Delta$ amount of load all on $s$.
We call the ``aggregated" result of such choices of agents of player~$i$ \emph{the strategy of player $i$},
and it can be represented by a $|\mc{S}_i|$-dimensional vector $x_i=(x_{i,s})_{s\in\mathcal{S}_i}$,
where $x_{i,s} \in [0,1]$ is the amount of the load that the agents of player $i$ puts on the path $s$.
Note that $\sum_{s\in\mathcal{S}_i}x_{i,s} =1/n$ and let $\mc{K}_i$ be the feasible set for all such vectors $x_i$.
Then the strategies of all players can be jointly represented by a vector
$$x=(x_1,...,x_n) = ((x_{1,s})_{s\in\mathcal{S}_1},...,(x_{n,s})_{s\in\mathcal{S}_n}) \in \mathbb{R}^d,$$
where $d=\sum_{i\in N}|\mathcal{S}_i|$, and let $\mc{K}=\mc{K}_1 \times \cdots \times \mc{K}_n$ be the feasible set for all such vectors $x$.
Each allowed path of a player intersects at most $k$ allowed paths (including that path itself) of that player.
We call $x_i$ the flow of player $i$ and $x$ the flow of the system.\footnote{Although we borrow the terms such as edge, path, and flow from routing games, the congestion games are more general as there are no underlying graphs and a path can be just any arbitrary subset of edges.}
Note that an edge $e \in E$ can be shared by different paths, and the aggregated load on $e$, denoted by $\ell_e(x)$, is $\sum_{s:e\in s}\sum_{i\in N}x_{i,s}$. The cost of a path $s$ is defined as $c_s(x)=\sum_{e\in s}c_e(\ell_e(x))$,
and the individual cost of player~$i$ is defined as $C_i(x)=\sum_{s\in\mathcal{S}_i}x_{i,s} c_s(x)$.
Each agent of player $i$ chooses a path $s$ from $\mc{S}_i$ that minimizes $c_s(x)$ in nonatomic congestion games.\footnote{One can show that our dynamics are no-regret algorithms for each agent of each player $i$, and this provides an incentive for the
agents to adopt the dynamics. If, instead of each agent choosing a path $s$ to minimize $c_s(x)$,
the agents of player $i$ cooperate to minimize the cost of player $i$, $C_i(x)$, we have an \emph{atomic splittable congestion game}.}


Such a game admits the following potential function (e.g., see \cite[Eq. (1)]{roughgarden2}, and it can be found as early as in \cite{beckmann:mcguire:winsten}):\footnote{Note that our convergence result will be proved more generally for any convex potential function satisfying certain properties.}
\begin{equation}\label{eq:po}
\Phi(x)=\sum_e\int_0^{\ell_e(x)}c_e(y)dy.
\end{equation}
To see that this is indeed a potential function, note that if some player deviates an infinitesimal fraction of load from $s$ to $s'$ (where $x_{i,s}>0$) such that $c_s(x)>c_{s'}(x')$ (where $x$ is almost the same as $x'$ except for the small fraction of moved load), then $\partial\Phi(x)/\partial x_{i,s}>\partial\Phi(x')/\partial x_{i,s'}$, which means that the rate of decrease in $\Phi$ is larger than the rate of increase in $\Phi$ and thus the resulting $\Phi$ decreases.
We will need the following, which we prove in Appendix~\ref{app:conv}.
\begin{proposition} \label{pro:conv}
The function $\Phi$ defined in (\ref{eq:po}) is convex.
\end{proposition}

\paragraph{Atomic Congestion Games.}
We later consider the \emph{atomic congestion game}, also described by $(N,E,(\mathcal{S}_i)_{i\in N},(c_e)_{e\in E})$.
Let us assume that there are $n$ players, each player has at most $d$ allowed paths, and each path has length at most $m$;
let $k$ be the maximum number of the allowed paths (including that path itself) of a player that each allowed path of that player intersects, and each player has a flow of amount $1/n$ to route.
The strategy of each player $i$ is to send her \emph{entire} flow on a single path, chosen randomly according to some distribution over her allowed paths, which can be represented by a $|\mc{S}_i|$-dimensional vector $\pi_i=(\pi_{i,s})_{s\in\mathcal{S}_i}$, where $\pi_{i,s} \in [0,1]$ is the probability of choosing path $s$.
It turns out to be more convenient for us to represent each player's strategy $\pi_i$ by an equivalent form $x_i = (1/n)\pi_i$, where $1/n$ is the amount of flow each player has. That is, for every $i \in N$ and $s \in \mc{S}_i$, $x_{i,s} = (1/n)\pi_{i,s} \in [0,1/n]$ and  $\sum_{s \in \mc{S}_i} x_{i,s} = 1/n$. Let $\mc{K}_i$ denote the feasible set of all such $x_i \in [0,1/n]^{|\mc{S}_i|}$ for player $i$, and let $\mc{K}=\mc{K}_1 \times \cdots \times \mc{K}_n$, which is the feasible set of all such joint strategy profiles $x=(x_1, \dots, x_n)$ of the $n$ players.

We will still be using the same function as in (\ref{eq:po}) for convergence analysis.
We are aware of the potential function of Rosenthal for atomic congestion games (see \cite{kleinberg:piliouras:tardos:multiplicative} for the potential function there), typically used for showing existence of pure Nash equilibria.
Actually, different functions could be used for different purposes.
For our purpose of doing mirror-gradient descents on some convex function to show actual convergence,
the function that we are using guarantees convexity (while the Rosenthal one does not) and thereby other convenience for analysis.

To represent which path a player $i$ actually chooses, we use another vector $X_i = (X_{i,s})_{s \in \mc{S}_i} \in \{0,1/n\}^{|\mc{S}_i|}$, where
\begin{equation}\label{eq:X}
X_{i,s} = \left\{
            \begin{array}{ll}
              1/n  & \hbox{if player $i$ chooses path $s$,} \\
              0 & \hbox{otherwise.}
            \end{array}
          \right.
\end{equation}
We call $X=(X_1, \dots, X_n)$ the \emph{choice vector} of the $n$ players. Then the cost of a path $s$ with respect to $X$ is defined as $$c_s(X) = \sum_{e \in s_i} c_e(\ell_e(X)),$$
where $\ell_e(X)$ is the amount of flow passing through edge $e$, defined as
\begin{equation}\label{eq:l_e}
\ell_e(X) = \sum_{j\in N} \sum_{r\in \mc{S}_j:e\in r} X_{j,r}.
\end{equation}
A useful property of the choice vector $X$ is that as each player $i$ chooses path $s$ with probability $\pi_{i,s} = x_{i,s} n$, the expected value of each $X_{i,s}$ is exactly $x_{i,s}$, which implies that $\Ex{X_i} = x_i$ for each $i$ and $\Ex{X} = x$.

\paragraph{Properties and Social Costs.}
As in \cite{kleinberg:piliouras:tardos:load}, we assume that the cost functions satisfy the property that for any $y\in [0,1]$ and any $e \in E$, $c_e(0)=0$, $c_e(1) \le 1$, $c'_e(y)\geq A>0$ and $0\leq c''_e(y)\leq B$, where $A,B$ are positive constants.\footnote{They are constants with respect to the amount of flow, not necessarily to $m$ or $n$.} By Lemma~4 of \cite{kleinberg:piliouras:tardos:load}, for constants $a=A$ and $b=B+1$ defined accordingly, the cost functions satisfy the condition that
\begin{equation}\label{eq:b}
a y\leq c_e(y)\leq  b y, \mbox{ for any }y\in[0,1].
\end{equation}

Then the function $\Phi$ defined in terms of such cost functions is smooth in the following sense.

\begin{definition} \label{def:smooth}
A function $\Phi$ over $\mc{K}$ is called $(\alpha,\beta,\lambda)$-smooth if for any $x \in \mc{K}$,
$$\Phi(x) \le \alpha, \|\nabla \Phi(x)\|_\infty \le \beta, \mbox{ and } \nabla^2 \Phi(x) \preceq \lambda I.$$
\end{definition}

\begin{proposition} \label{pro:para}
The function $\Phi$ defined in (\ref{eq:po}) with cost functions satisfying condition (\ref{eq:b}) is $(\alpha, \beta, \lambda)$-smooth, for
$$\alpha = bm /2, \beta = b m, \mbox{ and } \lambda = b m k.$$
\end{proposition}
We will prove the proposition in Appendix~\ref{app:para2}.

We consider two types of social cost functions. The first is the average individual cost function, defined as $$C_A(x)=\sum_e \ell_e(x) c_e(\ell_e(x)),$$
and the second is the maximum individual cost function, defined as
$$C_M(x)=\max_{s\in\mathcal{S}}\sum_{e\in s}c_e(\ell_e(x)), \mbox{ where } \mc{S}=\bigcup_{i \in N} \mc{S}_i.$$
 Using them, we measure the quality of outcome for a flow $x \in \mc{K}$ in the following two ways.
 The first is the ratio $C_A(x)/C_A(x^*)$, where $x^*=\arg\min_{z \in \mc{K}} C_A(z)$,
and the second is the ratio $C_M(x)/C_M(\hat{x})$, where $\hat{x}=\arg\min_{z \in \mc{K}} C_M(z)$.


\section{The Bulletin-Board Model} \label{sec:bulletin}
\subsection{Dynamics} \label{sec:con}


We consider the setting in which the players play the game iteratively in the following way. At step $t$, each player $i$ plays the strategy $\xx{x}{t}{i}$ by sending the amount $\xx{x}{t}{i,s}$ of load on path $s$ for each $s \in \mc{S}_i$. After that, she gets to know the vector $\xx{\hat c}{t}{i}=(c_s(\xx{x}{t}{}))_{s \in \mc{S}_i}$ of cost values, where $c_s(\xx{x}{t}{}) = \sum_{e \in s} c_e(\ell_e(\xx{x}{t}{}))$ is the cost value on the path $s$ at that step. With this, she updates her next strategy $\xx{x}{t+1}{i}$ in some way and then proceeds to the next iteration. In the alternative definition of the game, the corresponding setting is that at step $t$, each agent of player $i$ sends its load of $\Delta$ all on some path $s \in \mc{S}_i$, which is chosen according to some distribution. We assume that all agents of player $i$ start with the same initial distribution and update their distributions at each step $t$ using the same algorithm according to the same information $\xx{\hat c}{t}{i}$. Then we can conclude that their distributions at step $t$ are all the same,\footnote{The distributions of agents from different players are still different in general.} which basically can be described by the flow $\xx{x}{t}{i}$ of player $i$, due to the law of large number as the number of agents is huge. Thus, the settings for the two definitions of the game also match.

We have not specified how the players or agents of players update their next strategies. Different update algorithms may make the whole system evolve in rather different ways, and we would like to understand if there are update algorithms which players or agents of players have incentive to adopt that can lead to desirable outcomes for the whole system. One can argue that a plausible incentive for a player is to minimize her regret. Two well-known no-regret algorithms are the gradient descent algorithm and the multiplicative update algorithm, both of which can be seen as special cases of a more general algorithm called mirror descent algorithm (see e.g. \cite{BT03} for more detail). Inspired by this,
we consider the following update rule for player $i$ or agents of player $i$:
\begin{eqnarray}
\xx{x}{t+1}{i} &=& \arg\min_{z_i\in\mc{K}_i} \left\{\eta_i \langle \xx{\hat c}{t}{i}, z_i\rangle + \BRi{z_i,\xx{x}{t}{i}}\right\} \label{eq:update1}\\
&=& \arg\min_{z_i\in\mc{K}_i} \BRi{z_i,\xx{x}{t}{i}-\eta_i \xx{\hat c}{t}{i}}. \label{eq:update2}
\end{eqnarray}
Here, $\eta_i >0$ is some learning rate, $R_i:\mc{K}_i \to \mathbb{R}$ is some regularization function, and $\BRi{\cdot, \cdot}$ is the Bregman divergence with respect to $R_i$ defined as
$$\BRi{u_i,v_i} = R_i(u_i) - R_i(v_i) - \langle \nabla R_i(v_i), u_i-v_i\rangle$$
for $u_i,v_i \in \mc{K}_i$. This gives rise to a family of update rules for different choices of the function $R_i$. For example, it is well-known that by choosing $R_i(u_i) = \|u_i\|_2^2 /2$, one recovers the gradient descent algorithm, while by choosing $R_i(u_i) = \sum_s (u_{i,s} \ln u_{i,s} -u_{i,s})$, one recovers the multiplicative update algorithm. Using a similar argument as in \cite{kleinberg:piliouras:tardos:load}, one can show that this algorithm, with a properly chosen $R_i$, is indeed a no-regret algorithm for each agent of player $i$ (see Appendix~\ref{app:no-regret} for a proof sketch), and this provides an incentive for the agents to use the algorithm.
We need these $R_i$'s (and $\BRi{\cdot,\cdot}$'s) to satisfy the following.
The choices of $R_i$'s that satisfy this assumption will be discussed in Section~\ref{sec:conv}.

\begin{assumption} \label{as:cond}
For any $i \in N$ and any $x_i,y_i \in \mc{K}_i$,
$$\|x_i-y_i\|_2^2 \le 2 \cdot \BRi{x_i,y_i}.$$
\end{assumption}

Then the function $\Phi$ is ``smooth" with respect to these $R_i$'s in the following sense by Definition~\ref{def:smooth} and Assumption~\ref{as:cond}.
\begin{definition}
We say that $\Phi$ is $\lambda$-smooth with respect to $(R_1, \dots, R_n)$ if for any two inputs $x=(x_1, \dots, x_n)$ and $x'=(x'_1, \dots, x'_n)$ in $\mc{K}$,
\begin{equation}\label{eq:smooth}
\Phi(x') \le \Phi(x) + \langle \nabla \Phi(x), x'-x\rangle + \lambda \sum_{i=1}^n \BRi{x'_i,x_i}.
\end{equation}
\end{definition}


\subsection{Convergence Results and Equilibria}
Our main result in this section is the following, which shows that if each player (or agent of a player) uses such an update algorithm, the system quickly converges, in the sense that the value of the potential function $\Phi(\xx{x}{t}{})$ quickly approaches the minimum $\Phi(q)$, where $q=\arg\min_{z\in\mathcal{K}} \Phi(z)$.

\begin{theorem} \label{thm:Nash}
Consider any nonatomic congestion game of $n$ players,
with a potential function $\Phi$ which is $\lambda$-smooth with respect to some $(R_1,\dots,R_n)$. Let $q=(q_1, \dots, q_n) = \arg\min_{z\in\mathcal{K}} \Phi(z)$. Now suppose that each player $i$ starts from some initial strategy $\xx{x}{0}{i}$, with $\BRi{q_i,\xx{x}{0}{i}} \le \gamma$, and updates her strategy according to the rule in (\ref{eq:update1}), with $\eta_i \in[\eta, 1/\lambda]$ for some $\eta$. Then for 
any $\varepsilon \in (0,1)$ there exists some $T_{\varepsilon} \le n\gamma/(\eta \varepsilon)$ such that for any $t \ge T_{\varepsilon}$, $\Phi(\xx{x}{t}{}) \le \Phi(q) + \varepsilon$.
\end{theorem}


We will prove this main result right after presenting its derived result.
From Theorem~\ref{thm:Nash}, we have the following, which we will prove in Section~\ref{sec:conv}.

\begin{corollary} \label{cor:conv}
Consider any nonatomic congestion game of $n$ players with parameters given in Section~\ref{sec:preliminaries}, and let $\lambda=mb k$. Now if each player $i$ plays the gradient descent algorithm by starting from any $\xx{x}{0}{i} \in \mc{K}_i$ and using any $\eta_i\in [\eta, 1/\lambda]$, then $T_{\varepsilon} \le 2/(n\eta \varepsilon)$. Furthermore, if each player $i$ plays the multiplicative update algorithm by starting from a uniform $\xx{x}{0}{i}$ (same load on each allowed path) and using any $\eta_i\in [\eta, 1/\lambda]$, then $T_{\varepsilon} \le (n \ln (dn))/(\eta \varepsilon)$.
\end{corollary}

\begin{remark} \label{rem:eta}
According to Corollary~\ref{cor:conv}, playing the gradient descent algorithm guarantees a faster convergence time. In particular, if each player $i$ uses $\eta_i=1/\lambda$, then adopting the gradient descent algorithm leads to a convergence time $T_{\varepsilon} \le 2m b k /(n\varepsilon)$, while adopting the multiplicative update algorithm leads to $T_{\varepsilon} \le (m b k n \ln (dn)) /\varepsilon$.
\end{remark}

Note that the given bound on $T_{\varepsilon}$ is proportional to $\gamma$ by Theorem~\ref{thm:Nash},
which is an upper bound on $\BRi{q_i,\xx{x}{0}{i}}$.
In Section~\ref{sec:conv} of Proof of Corollary~\ref{cor:conv}, we can see that $\gamma$ is smaller in the gradient descent algorithm than in the multiplicative update algorithm.
Thus, the corresponding different bounds on $T_{\varepsilon}$.

Implications of $\Phi(\xx{x}{t}{})$ being close to $\Phi(q)$ include $\xx{x}{t}{}$ being an approximate equilibrium,
which will be proved in Section~\ref{sec:approximate}, and achieving social efficiency, which will be given in Section~\ref{sec:splittable}.
We say that a flow $x \in \mc{K}$ is an $\delta$-equilibrium if for any player $i \in N$ and any paths $s, s' \in \mc{S}_i$ with $ x_{i,s}>0$, $c_s(x) \le c_{s'}(x) + \delta$. Note that with $\delta=0$, we recover the standard definition of equilibrium for nonatomic games. The following shows that after the convergence time, the system playing our algorithm will stay in an $\delta$-equilibrium for a small $\delta$.
\begin{theorem} \label{thm:diff}
Any $x \in \mc{K}$ such that $\Phi(x)\leq\Phi(q)+\varepsilon$ for any $\varepsilon$ must be a $\delta$-equilibrium for some $\delta \le \sqrt{8 b m\varepsilon}$.
\end{theorem}

\subsubsection*{Analysis}
To prove Theorem~\ref{thm:Nash}, the key observation is that the updates by all players collectively can be seen as doing a generalized version of the mirror descent, with different step sizes in different dimensions, on the potential function $\Phi$ defined in (\ref{eq:po}). To see this, note that for any $i \in N$ and $s \in \mc{S}_i$, the $s$'th entry of $\xx{\hat c}{t}{i}$ is
$$c_s(\xx{x}{t}{}) = \sum_{e \in s} c_e(\ell_e(\xx{x}{t}{})) = \frac{\partial \Phi(\xx{x}{t}{})}{\partial x_{i,s}},$$
which means that the $d$-dimensional vector $(\xx{\hat c}{t}{i})_{i \in N}$ is in fact equal to $\nabla\Phi(\xx{x}{t}{})$, the gradient of $\Phi$ at $\xx{x}{t}{}$.
That is, if we write $\nabla\Phi(\xx{x}{t}{})=(\nabla_1\Phi(\xx{x}{t}{}), \dots, \nabla_n\Phi(\xx{x}{t}{}))$, with $\nabla_i \Phi(\xx{x}{t}{})$ being the portion of $\nabla \Phi(\xx{x}{t}{})$ corresponding to player $i$, then the update rule of (\ref{eq:update1}) and (\ref{eq:update2}) becomes the following:
\begin{eqnarray}
\xx{x}{t+1}{i} &=& \arg\min_{z_i\in\mc{K}_i} \left\{\eta_i \langle \nabla_i \Phi(\xx{x}{t}{}), z_i\rangle + \BRi{z_i,\xx{x}{t}{i}}\right\} \label{eq:GD}\\
&=& \arg\min_{z_i\in\mc{K}_i} \BRi{z_i,\xx{x}{t}{i}-\eta_i \nabla_i \Phi(\xx{x}{t}{})}.
\end{eqnarray}
Observe that when all the $\eta_i$'s are identical, the collective update of all players moves the whole system exactly in the direction of $-\nabla \Phi(\xx{x}{t}{})$, and this becomes the standard mirror descent algorithm which has the same step size across all dimensions. It is known that doing such a mirror descent on a smooth convex function leads to a fast convergence to its minimum \cite{BDX11,nesterov}. On the other hand, we consider the more general case in which different players can have different learning rates, and this corresponds to a more general mirror descent algorithm which allows different step sizes in different dimensions. Because the different step sizes have different scaling effects in different dimensions, the collective update now no longer moves the whole system in the direction of $-\nabla \Phi(\xx{x}{t}{})$, and it is not clear if a similar convergence result can be obtained. Interestingly, the following theorem shows that doing such a generalized mirror descent algorithm on a general smooth convex function still gives us a fast convergence to its minimum.

\begin{theorem} \label{thm:GD}
Suppose $\mc{K}=\mc{K}_1 \times \cdots \times \mc{K}_n$, with each $\mc{K}_i$ being a convex set. Let $\Phi: \mc{K}\rightarrow\mathbb{R}$ be any convex function which is $\lambda$-smooth with respect to some $(R_1,\dots,R_n)$ and let $q=(q_1, \dots, q_n)= \arg\min_{z\in\mathcal{K}} \Phi(z)$. Suppose we start from some $\xx{x}{0}{}=(\xx{x}{0}{1},\dots,\xx{x}{0}{n})$, with each $\BRi{q_i,\xx{x}{0}{i}} \le \gamma$, and then use the update rule in (\ref{eq:GD}), with each $\eta_i \in [\eta, 1/\lambda]$ for some $\eta$. Then for any $\varepsilon \in (0,1)$, there exists some $T_{\varepsilon} \le n\gamma/(\eta \varepsilon)$ such that for any $t \ge T_{\varepsilon}$, $\Phi(\xx{x}{t}{}) \le \Phi(q) + \varepsilon$.
\end{theorem}


We will prove Theorem~\ref{thm:GD} in Section~\ref{sec:GD}. Now note that Theorem~\ref{thm:Nash} follows immediately from Theorem~\ref{thm:GD} since our potential function $\Phi$ is convex by Proposition~\ref{pro:conv}. On the other hand, Theorem~\ref{thm:GD} works for a general convex function (not restricted to the specific potential function given in (\ref{eq:po})), which may have independent interest of its own.


\subsubsection{Proof of Theorem~\ref{thm:GD}} \label{sec:GD}

Our proof follows closely that in \cite{BDX11} for the special case in which all the $\eta_i$'s are identical. To simplify our notation, let us denote the gradient vector $\nabla \Phi(\xx{x}{t}{})$ by $\xx{g}{t}{}=(\xx{g}{t}{1}, \dots, \xx{g}{t}{n})$, with $\xx{g}{t}{i}=\nabla_i \Phi(\xx{x}{t}{})$.

Using the assumption that for each $i$, $\eta_i \le 1/\lambda$ and thus $\lambda \le 1/\eta_i$, the $\lambda$-smoothness condition~(\ref{eq:smooth}) implies that
\begin{equation}\label{eq:eta}
\Phi(\xx{x}{t+1}{}) \le \Phi(\xx{x}{t}{}) + \langle \xx{g}{t}{}, \xx{x}{t+1}{}-\xx{x}{t}{}\rangle + \sum_{i=1}^n \frac{1}{\eta_i} \BRi{\xx{x}{t+1}{i},\xx{x}{t}{i}},
\end{equation}
because each $\BRi{\xx{x_i}{t+1}{},\xx{x_i}{t}{}}$ is nonnegative.
Then we need the following two lemmas, which we will prove later.
\begin{lemma} \label{lem:less}
For any integer $t \ge 0$, $\Phi(\xx{x}{t+1}{}) \le \Phi(\xx{x}{t}{})$.
\end{lemma}

\begin{lemma} \label{lem:sum}
For any integer $T \ge 1$,
$$\sum_{t=0}^{T-1} \left(\Phi(\xx{x}{t+1}{}) - \Phi(q)\right) \le \sum_{i=1}^n  \frac{1}{\eta_i} \BRi{q_i,\xx{x}{0}{i}}.$$
\end{lemma}

Combining these two lemmas together, we obtain
\begin{eqnarray*}
T \left(\Phi(\xx{x}{T}{}) - \Phi(q)\right) &\le& \sum_{t=0}^{T-1} \left(\Phi(\xx{x}{t+1}{}) - \Phi(q)\right)\\
&\le& \sum_{i=1}^n  \frac{1}{\eta_i} \BRi{q_i,\xx{x}{0}{i}}\\
&\le& \frac{n \gamma}{\eta}.
\end{eqnarray*}
Dividing both sides by $T$ gives us
$$\Phi(\xx{x}{T}{}) - \Phi(q) \le \frac{n \gamma}{\eta T} \le \varepsilon,$$
when $T \ge n \gamma/(\eta \varepsilon)$, and we have the theorem. It remains to prove the two lemmas, which we do next.


\begin{proof}[Proof of Lemma~\ref{lem:less}]
We know from (\ref{eq:eta}) that
$$\Phi(\xx{x}{t+1}{}) \le \Phi(\xx{x}{t}{}) + \sum_{i=1}^n  \left(\langle \xx{g}{t}{i}, \xx{x}{t+1}{i}-\xx{x}{t}{i}\rangle + \frac{1}{\eta_i} \BRi{\xx{x}{t+1}{i},\xx{x}{t}{i}}\right).$$
To bound the sum above, note that according to the definition of $\xx{x}{t+1}{i}$ in (\ref{eq:GD}), we have
\begin{eqnarray*}
\lefteqn{\langle \xx{g}{t}{i}, \xx{x}{t+1}{i}-\xx{x}{t}{i}\rangle + \frac{1}{\eta_i} \BRi{\xx{x}{t+1}{i},\xx{x}{t}{i}}}\\
&\le& \langle \xx{g}{t}{i}, \xx{x}{t}{i}-\xx{x}{t}{i}\rangle + \frac{1}{\eta_i} \BRi{\xx{x}{t}{i},\xx{x}{t}{i}}\\
&=& 0.
\end{eqnarray*}
Applying this to the above bound on $\Phi(\xx{x}{t+1}{})$, Lemma~\ref{lem:less} follows.
\end{proof}

\begin{proof}[Proof of Lemma~\ref{lem:sum}]
We know from (\ref{eq:eta}) that for any $t \ge 0$, $\Phi(\xx{x}{t+1}{})$ is at most
$$\Phi(\xx{x}{t}{}) + \langle \xx{g}{t}{}, \xx{x}{t+1}{}-\xx{x}{t}{}\rangle + \sum_{i=1}^n \frac{1}{\eta_i} \BRi{\xx{x}{t+1}{i},\xx{x}{t}{i}},$$
where the second term above can be expressed as
\begin{eqnarray*}
\langle \xx{g}{t}{}, \xx{x}{t+1}{}-\xx{x}{t}{}\rangle &=& \langle \xx{g}{t}{}, q-\xx{x}{t}{}\rangle + \langle \xx{g}{t}{}, \xx{x}{t+1}{}-q\rangle\\
&=& \langle \xx{g}{t}{}, q-\xx{x}{t}{}\rangle + \sum_{i=1}^n \langle \xx{g}{t}{i}, \xx{x}{t+1}{i}-q_i\rangle.
\end{eqnarray*}
Since $\Phi(\xx{x}{t}{}) + \langle \xx{g}{t}{}, q-\xx{x}{t}{}\rangle \le \Phi(q)$ for a convex $\Phi$, we thus know that $\Phi(\xx{x}{t+1}{})$ is at most
\begin{equation}\label{eq:up}
\Phi(q) + \sum_{i=1}^n  \left(\langle \xx{g}{t}{i}, \xx{x}{t+1}{i}-q_i\rangle + \frac{1}{\eta_i} \BRi{\xx{x}{t+1}{},\xx{x}{t}{}}\right).
\end{equation}
To bound the sum above, we rely on the following.
\begin{proposition}
For each $i$,
$\langle \xx{g}{t}{i}, \xx{x}{t+1}{i}-q_i\rangle$ is at most
$$\frac{1}{\eta_i} \left(\BRi{q_i,\xx{x}{t}{i}} - \BRi{q_i,\xx{x}{t+1}{i}} - \BRi{\xx{x}{t+1}{i},\xx{x}{t}{i}}\right).$$
\end{proposition}
\begin{proof}
According to the definition of $\xx{x}{t+1}{i}$ in (\ref{eq:GD}), it is also the minimizer of the function
$$L(z)= \eta_i \langle \xx{g}{t}{i}, z-q_i\rangle + \BRi{z,\xx{x}{t}{i}}$$
over $z \in \mc{K}_i$, since $\langle \xx{g}{t}{i}, -q_i\rangle$ is a constant independent of $z$. Then from a well-known fact in convex optimization \cite[p.139-140]{boyd:vandenberghe}, we know that
$$\langle \nabla L(\xx{x}{t+1}{i}), q_i - \xx{x}{t+1}{i} \rangle \ge 0.$$
Since $\nabla L(\xx{x}{t+1}{i}) = \eta_i \xx{g}{t}{i} + \nabla R_i(\xx{x}{t+1}{i}) - \nabla R_i(\xx{x}{t}{i})$, we have
\begin{equation}\label{eq:BR}
\eta_i \langle \xx{g}{t}{i}, \xx{x}{t+1}{i}-q_i \rangle \le \left\langle \nabla R_i(\xx{x}{t+1}{i}) - \nabla R_i(\xx{x}{t}{i}), q_i - \xx{x}{t+1}{i} \right\rangle.
\end{equation}
Then according to the definition of $\BRi{\cdot}$, we have
\begin{eqnarray*}
\lefteqn{\BRi{q_i,\xx{x}{t}{i}}}\\ &=& R_i(q_i) - R_i(\xx{x}{t}{i}) - \langle \nabla R_i(\xx{x}{t}{i}), q_i - \xx{x}{t}{i}\rangle,\\
\lefteqn{\BRi{q_i,\xx{x}{t+1}{i}}}\\ &=& R_i(q_i) - R_i(\xx{x}{t+1}{i}) - \langle \nabla R_i(\xx{x}{t+1}{i}), q_i - \xx{x}{t+1}{i}\rangle, \mbox{ and}\\
\lefteqn{\BRi{\xx{x}{t+1}{i},\xx{x}{t}{i}}}\\ &=& R_i(\xx{x}{t+1}{i}) - R_i(\xx{x}{t}{i}) - \langle \nabla R_i(\xx{x}{t}{i}), \xx{x}{t+1}{i} - \xx{x}{t}{i}\rangle.
\end{eqnarray*}
By subtracting the second and the third equalities from the first, we obtain
\begin{eqnarray*}
\lefteqn{\BRi{q_i,\xx{x}{t}{i}} - \BRi{q_i,\xx{x}{t+1}{i}} - \BRi{\xx{x}{t+1}{i},\xx{x}{t}{i}}}\\
&=& \left\langle \nabla R_i(\xx{x}{t+1}{i}) - \nabla R_i(\xx{x}{t}{i}), q_i - \xx{x}{t+1}{i} \right\rangle.
\end{eqnarray*}
Substituting this into (\ref{eq:BR}) proves the proposition.
\end{proof}

Combining the bound from this proposition with the upper bound on $\Phi(\xx{x}{t+1}{})$ in (\ref{eq:up}), we obtain
$$\Phi(\xx{x}{t+1}{}) \le \Phi(q) + \sum_{i=1}^n \frac{1}{\eta_i} \left(\BRi{q_i,\xx{x}{t}{i}} - \BRi{q_i,\xx{x}{t+1}{i}}\right).$$
This implies that
\begin{eqnarray*}
\lefteqn{\sum_{t=0}^{T-1} \left(\Phi(\xx{x}{t+1}{}) - \Phi(q)\right)}\\
&\le& \sum_{i=1}^n  \frac{1}{\eta_i} \sum_{t=0}^{T-1} \left(\BRi{q,\xx{x}{t}{i}} - \BRi{q,\xx{x}{t+1}{i}}\right)\\
&\le& \sum_{i=1}^n  \frac{1}{\eta_i} \BRi{q_i,\xx{x}{0}{i}},
\end{eqnarray*}
which proves Lemma~\ref{lem:sum}.
\end{proof}

\subsubsection{Proof of Corollary~\ref{cor:conv}} \label{sec:conv}

Let us first consider the case that each player plays the gradient descent algorithm. Note that this corresponds to choosing $R_i(u_i) = \| u_i \|_2^2 /2$ for each $i$, and one can show that $\BRi{u_i,v_i} = \|u_i - v_i\|^2_2 /2$, for $u_i,v_i \in \mc{K}_i$. Then, we have
$$\BRi{q_i,\xx{x}{0}{i}} = \|q_i - \xx{x}{0}{i}\|^2_2 /2 \le \|q_i - \xx{x}{0}{i}\|_1^2 /2$$
which is at most $$\left(\|q_i\|_1 + \|\xx{x}{0}{i}\|_1\right)^2 /2 \le 2/n^2.$$
Therefore, we can choose $\gamma=2/n^2$ to have $\BRi{q_i,\xx{x}{0}{i}} \le \gamma$. Furthermore, using the Taylor expansion together with Proposition~\ref{pro:para}, we know that for any $x,x' \in \mc{K}$,
$$\Phi(x') \le \Phi(x)+\langle\nabla \Phi(x), x'-x\rangle + \lambda \|x'-x\|_2^2 / 2,$$
with $\lambda=mb k$. Since
$$\|x'-x\|_2^2 /2 = \sum_i \|x'_i-x_i\|_2^2 /2= \sum_i \BRi{x'_i,x_i},$$
we can guarantee that $\Phi$ is $\lambda$-smooth with this choice of $R_i$'s.

Next, let us consider the case that each player plays the multiplicative update algorithm.
Note that this corresponds to choosing $R_i(u_i) = \sum_s (u_{i,s} \ln u_{i,s} - u_{i,s})$ for each $i$,
and one can show that $\BRi{u_i,v_i} = \sum_s u_{i,s} \ln (u_{i,s}/v_{i,s})$, for $u_i,v_i \in \mc{K}_i$. Then, we have
$$\BRi{q_i,\xx{x}{0}{i}} \le \sum_s q_{i,s} \ln (|\mc{S}_i|n) \le \ln (dn).$$
Therefore, we can choose $\gamma=\ln (dn)$ to have $\BRi{q_i,\xx{x}{0}{i}} \le \gamma$. Furthermore, we know that $$\|x'_i-x_i\|_2^2 /2 \le \|x'_i - x_i\|_1^2 /2 \le \BRi{x'_i,x_i},$$
by Pinsker's inequality.\footnote{Pinsker's inequality states that the total variance is upper bounded by the KL-divergence between two probability distributions where the total variance can be defined as half the 1-norm between these two distributions. The Bregman divergence is KL-divergence here.} Therefore, we can again guarantee that $\Phi$ is $\lambda$-smooth with this choice of $R_i$'s.

Substituting these bounds of $\gamma$ and $\lambda$ into Theorem~\ref{thm:Nash}, Corollary~\ref{cor:conv} then follows.

\subsubsection{Proof of Theorem~\ref{thm:diff}} \label{sec:approximate}

Consider any $x \in \mc{K}$ such that $\Phi(x)\leq\Phi(q)+\varepsilon$ and any $i \in N$. Let $s_0$ be the path in $\mc{S}_i$ which minimizes $c_s(x)$ among $s \in \mc{S}_i$, and let $s_1$ be the path which maximizes $c_s(x)$ among $s \in \mc{S}_i$ with $ x_{i,s}>0$. Let $\delta = c_{s_1}(x) - c_{s_0}(x)$ and our goal is to show that $\delta$ is small by bounding it in terms of $\varepsilon$. The idea is that if $\delta$ were large, we could move a significant amount of load from $s_1$ to $s_0$ and decrease the $\Phi$ value substantially, which is impossible as $\Phi(x)$ is close to the minimum value $\Phi(q)$.

Formally, let us move some $\Delta$ amount, which we will set later, of load from $s_1$ to $s_0$, and let $z$ denote the new flow.
Note that the cost increase on $s_0$ and the cost decrease on $s_1$ are both at most $m b \Delta$, since $c_e'(y) \le  b $ for any $y$ according to the condition (\ref{eq:b}).
Thus, with $\Delta=\delta/(4 b m)$ (since $\delta=4 m b \Delta$, i.e., the path cost difference can be expressed by multiplying the maximum number of edges in a path, the upper bound on an edge cost changing rate, the amount of load moved, and some proper constant scaling), we can have $c_{s_1}(z) \ge c_{s_1}(x) - \delta/4$ and $c_{s_0}(z) \le c_{s_0}(x) + \delta/4$, so that
\begin{eqnarray} \label{eqn:3}
c_{s_1}(z) - c_{s_0}(z) \ge
c_{s_1}(x) - c_{s_0}(x) - \delta/2 = \delta/2,
\end{eqnarray}
which can be used to bound the decrease of the $\Phi$ value.

Moving the load decreases the $\Phi$ value by the amount
\begin{eqnarray*}
\lefteqn{\Phi(x)-\Phi(z)}\\
&=& \sum_{e \in s_1 \setminus s_0} \int_{\ell_e(x)-\Delta}^{\ell_e(x)} c_e(y) dy - \sum_{e \in s_0 \setminus s_1} \int_{\ell_e(x)}^{\ell_e(x)+\Delta} c_e(y) dy\\
&\ge& \sum_{e \in s_1 \setminus s_0} \Delta c_e(\ell_e(x)-\Delta) - \sum_{e \in s_0 \setminus s_1} \Delta c_e(\ell_e(x)+\Delta)\\
&=& \Delta \sum_{e\in s_1} c_e(\ell_e(z)) - \Delta \sum_{e\in s_0} c_e(\ell_e(z))\\
&=& \Delta \left(c_{s_1}(z) - c_{s_0}(z)\right)\\
&\ge& \Delta\delta/2,
\end{eqnarray*}
where the first inequality holds as the function $c_e$ is nondecreasing and the last inequality holds by (\ref{eqn:3}). Since $z$ is still a feasible flow in $\mc{K}$, its $\Phi$ value cannot be smaller than that of $q$ and we must have $\Phi(x)-\Phi(z) \le \Phi(x)-\Phi(q) \le \varepsilon$, which implies that $\Delta \delta/2 \le \varepsilon$. With $\Delta = \delta/(4 b m)$, we have $\delta \le \sqrt{8 b m\varepsilon}$ (and $\Delta\leq\sqrt{\frac{\varepsilon}{2bm}}$). Since this holds for any $i \in N$, we have the theorem.

\subsection{Social Costs} \label{sec:splittable}

According to Theorem~\ref{thm:Nash}, the flow $\xx{x}{t}{}$ at step $t \ge T_{\varepsilon}$ enjoys the nice property that $\Phi(\xx{x}{t}{}) \le \Phi(q) + \varepsilon$. In this section, we show the implication of this property on the social costs.

\subsubsection{Average Individual Cost}

We show that after the convergence time, the average individual cost achieved by our algorithm is only within a constant factor from the optimum one.

\begin{theorem} \label{thm:average}
Any $x \in \mc{K}$ such that $\Phi(x)\leq\Phi(q)+\varepsilon$ must have $\frac{C_A(x)}{C_A(x^*)}\leq\frac{ b }{a}\left(1+\frac{2m\varepsilon}{a}\right).$
\end{theorem}
\begin{proof}
For any $z \in \mc{K}$, we can rewrite $C_A(z)$ as
\begin{eqnarray*}
C_A(z) &=& \sum_e \ell_e(z) c_e(\ell_e(z))\\
&=& \sum_e\int_0^{\ell_e(z)} \left(y c_e(y)\right)' dy \\
&=& \sum_e\int_0^{\ell_e(z)} \left(c_e(y)+yc'_e(y) \right) dy.
\end{eqnarray*}
Under the condition (\ref{eq:b}), we have $yc'_e(y) \le y b  = \frac{ b }{a} b_0 y \le \frac{ b }{a} c_e(y)$ and thus
\begin{eqnarray}
C_A(z) &\le& \sum_e\int_0^{\ell_e(z)} \left(1+\frac{ b }{b_0} \right) c_e(y) dy \nonumber\\
&=& \frac{a+b}{a} \Phi(z).\label{eq:le}
\end{eqnarray}
On the other hand, we also have
$yc'_e(y) \ge y b_0 = \frac{a}{ b }  b y \ge \frac{a}{ b } c_e(y)$ and thus
\begin{eqnarray}
C_A(z) &\ge& \sum_e\int_0^{\ell_e(z)} \left(1+\frac{a}{ b } \right) c_e(y) dy \nonumber\\
&=& \frac{a+b}{b} \Phi(z).\label{eq:ge}
\end{eqnarray}
Replacing $z$ in (\ref{eq:le}) by $x$ with $\Phi(x) \le \Phi(q)+\varepsilon$, and replacing $z$ in (\ref{eq:ge}) by $x^*$, we obtain
$$\frac{C_A(x)}{C_A(x^*)} \le \frac{ b }{a} \frac{\Phi(x)}{\Phi(x^*)} \le \frac{ b }{a} \frac{\Phi(q)+\varepsilon}{\Phi(q)},$$
as $\Phi(x^*) \ge \Phi(q)$, which gives us
\begin{equation}\label{eq:C_A}
\frac{C_A(x)}{C_A(x^*)} \le \frac{ b }{a} \left(1+\frac{\varepsilon}{\Phi(q)}\right).
\end{equation}
Finally, using the condition (\ref{eq:b}), we have for any $z \in \mc{K}$ that
\begin{eqnarray}
\Phi(z) &\ge& \sum_e\int_0^{\ell_e(z)} a y dy = \frac{a}{2} \sum_e (\ell_e(z))^2 \notag\\
 &\ge& \frac{a}{2m} \left(\sum_e \ell_e(z)\right)^2 \ge \frac{a}{2m}, \label{eq:Phi}
\end{eqnarray}
where the second inequality is by Cauchy-Schwarz and the last inequality holds as the total load of players is $1$. Substituting the bound of (\ref{eq:Phi}) into (\ref{eq:C_A}) with $z=q$, we have the theorem.
\end{proof}

\begin{remark}
We can make $\frac{C_A(x)}{C_A(x^*)}\leq\frac{ b }{a}\left(1+\sigma\right)$ for any $\sigma$ we want, by choosing $\varepsilon = a \sigma/(2m)$. Then by Remark~\ref{rem:eta}, one can compute the corresponding convergence time $T_\varepsilon$, which is proportional to $1/\sigma$.
\end{remark}

\subsubsection{Maximum Individual Cost in Symmetric Games}

In a symmetric game, $\mathcal{S}_i=\mathcal{S}$ for every $i\in N$. Taking advantage of this property, we show that after the convergence time the maximum individual cost achieved by our algorithm is again within a constant factor from the optimum one.
\begin{theorem}
Any $x \in \mc{K}$ such that $\Phi(x)\leq\Phi(q)+\varepsilon$ 
must have $\frac{C_M(x)}{C_M(\hat{x})} \leq \frac{b}{a}\left(1+\frac{2m\varepsilon}{a}+\frac{\delta m}{b}\right),$ where $\delta \le \sqrt{8 b m\varepsilon}$.
\end{theorem}
\begin{proof}
Consider any $x \in \mc{K}$ with $\Phi(x)\leq\Phi(q)+\varepsilon$. Let $s_0 =\arg\min_{s \in \mc{S}} c_s(x)$ and $s_1 =\arg\max_{s \in \mc{S}} c_s(x)$. To apply Theorem~\ref{thm:diff}, let us choose a player $i$ with $x_{i,s_1}>0$; such a player must exist because otherwise there would be no load on $s_1$ and $c_{s_1}(x)=0$ could not be the highest path cost. Since $\mathcal{S}_i=\mathcal{S}$ in a symmetric game, $s_0$ is also the path of player $i$ with the lowest path cost. Therefore, we can apply Theorem~\ref{thm:diff} and have $\delta = c_{s_1}(x) - c_{s_0}(x) \le \sqrt{8 b m\varepsilon}$. Note that $C_M(x)=c_{s_1}(x)$ by definition.
Thus, we have
$$\frac{C_M(x)}{C_M(\hat{x})} \le \frac{c_{s_1}(x)}{C_A(\hat{x})} = \frac{c_{s_0}(x)+\delta}{C_A(\hat{x})} \le \frac{C_A(x)+\delta}{C_A(x^*)},$$
where the first inequality is by the definitions of $C_M$ and $C_A$, and the second inequality follows from the fact that $c_{s_0}(x)\le C_A(x)$ and $x^*$ minimizes $C_A$. Furthermore,
$$\frac{C_A(x)+\delta}{C_A(x^*)} = \frac{C_A(x)}{C_A(x^*)}+\frac{\delta}{C_A(x^*)} \le \frac{b}{a}\left(1+\frac{2\varepsilon m}{a}\right) +\frac{\delta}{C_A(x^*)}$$
by Theorem~\ref{thm:average}. Finally, using a similar analysis as in the proof of Theorem~\ref{thm:average}, one can show that
$$C_A(x^*) \ge \sum_e\int_0^{\ell_e(x^*)} (a y + a y) dy = a \sum_e (\ell_e(x^*))^2 \ge \frac{a}{m}.$$
Combining all the bounds together, we have the theorem.
\end{proof}

\begin{remark}
We can make $\frac{C_M(x)}{C_M(\hat{x})} \leq \frac{b}{a}(1+\sigma)$ for any $\sigma$ we want, by choosing $\varepsilon = a \sigma^2 / (32m)$. Then according to Remark~\ref{rem:eta}, one can compute the corresponding convergence time $T_\varepsilon$, which is now proportional to $1/\sigma^2$.
\end{remark}

\section{The Bandit Model} \label{sec:bandit}

We consider the setting in which the players play the congestion game iteratively in the following way. In step $t$, each player $i$ plays some strategy $\xx{x}{t}{i}$ by sampling a path $s_i \in \mc{S}_i$ with probability $\xx{\pi}{t}{i,s_i} = \xx{x}{t}{i,s_i} n$ and sending her entire flow of amount $1/n$ through that path $s_i$. After that, she gets to know the cost of that path, which is
$$c_{s_i}(\xx{X}{t}{}) = \sum_{e \in s_i} c_e(\ell_e(\xx{X}{t}{})),$$
where $\xx{X}{t}{}$ is the choice vector of the players in step $t$.
With this feedback, she updates her next strategy $\xx{x}{t+1}{i}$
in some way and then proceeds to the next step. Let us stress that the only information a player has access to in a step is the cost of the path she has just chosen, and she has no information about that of any other path she has not chosen. This is known as the \emph{bandit} model in the area of online learning \cite[Chapter 4]{NRTV07}\cite{auer:cesa-bianchi,flaxman:kalai}. It is a more challenging model compared to the \emph{bulletin-board} model considered in \cite{kleinberg:piliouras:tardos:load} and in Section~\ref{sec:bulletin}, in which each player $i$ can learn the costs of all her allowed (even unchosen) paths.

We consider the same function for the convergence purpose:
\begin{equation}\label{eq:po2}
\Phi(x)=\sum_{e\in E} \int_0^{\ell_e(x)} c_e(y)dy,
\end{equation}
for $x \in \mc{K}$, where
$$\ell_e(x) = \sum_{j\in N} \sum_{r\in \mc{S}_j:e\in r} \xx{x}{t}{j,r},$$
with $\ell(\cdot)$ being the same function defined in (\ref{eq:l_e}).
According to Proposition~\ref{pro:conv}, $\Phi$ is a convex function.
We still assume that the cost functions satisfy the property in (\ref{eq:b}).


Note that this function $\Phi$ was used in Section~\ref{sec:bulletin} for a nonatomic version of the congestion game, with its input corresponding to splittable flows of players, and its minimizer corresponding to a pure Nash equilibrium with infinitely many agents (i.e., a Wardrop equilibrium). As we consider the atomic version of the congestion game, we now interpret its input as the mixed strategies of players. However, its minimizer corresponds no longer to a pure or mixed Nash equilibrium, but instead to an approximate Nash equilibrium, to be discussed later. We will rely on the following proposition, whose proof we omit here as it follows easily from an analysis in \cite[Proof of Lemma 8]{kleinberg:piliouras:tardos:load} using Taylor expansion.
The proposition intuitively means that the difference between the expected cost and the cost of expected flow (proportional to the probability distribution of mixed strategies) on $s$ is bounded.

\begin{proposition} \label{pro:exp}
Let $x \in \mc{K}$ and $X$ the choice vector sampled according to $x$. Then for any $s$, $0 \le \Ex{c_s(X)} - c_s(x) \le \frac{Bm}{8n}$
where $B$ is the bound on the second derivative of cost functions (Section~\ref{sec:preliminaries}).
This implies that classes of cost functions with smaller $B$ give better bounds on the difference.
In particular, with linear cost functions where $B=0$, then $\Ex{c_s(X)} = c_s(x)$ for any $s$.
\end{proposition}

\subsection{Bandit Algorithms} \label{sec:alg}

We would like to follow the approach of Section~\ref{sec:con} which showed that when each player plays a mirror-descent algorithm with her own learning rate,
the joint strategy profile of the players converges quickly.
However, there are two key differences in our setting which prevent us from applying their result directly. The first is that here we consider the atomic version of the congestion game, in which each player must send her entire flow on one single path, unlike the non-atomic version considered in Section~\ref{sec:bulletin} in which a player can split her flow across multiple paths in each step. The second difference is that the feedback model considered in Section~\ref{sec:bulletin} is the easier bulletin model, in which each player, after sending her flow in a step, gets to see all the costs of her allowed paths. That is, in step $t$, player $i$ is able to learn the cost
$$c_s(\xx{x}{t}{}) = \sum_{e \in s} c_e(\ell_e(\xx{x}{t}{})), \;\mbox{ with } \ell_e(x) = \sum_{j\in N} \sum_{r\in \mc{S}_j:e\in r} \xx{x}{t}{j,r},$$
for every $s \in \mc{S}_i$, where $\xx{x}{t}{}$ is the joint strategy profile of players, describing how each player splits her flow.
It was shown in Section~\ref{sec:con} that for any $i \in N$ and $s \in \mc{S}_i$,
\begin{equation}\label{eq:gra}
c_s(\xx{x}{t}{}) = \sum_{e \in s} c_e(\ell_e(\xx{x}{t}{})) = \frac{\partial \Phi(\xx{x}{t}{})}{\partial x_{i,s}},
\end{equation}
which means that each player $i$ gets to know the part of the gradient vector $\nabla \Phi(\xx{x}{t}{})$ related to her, denoted as $\nabla_i \Phi(\xx{x}{t}{})$, consisting of $\frac{\partial \Phi(\xx{x}{t}{})}{\partial x_{i,s}}$ for $s \in \mc{S}_i$. This allows each player to update her strategy in a way that the whole system of players can be seen as collectively running a generalized mirror-descent algorithm on the convex potential function $\Phi$. However, in the more challenging bandit model we adopt in this paper, each player $i$ does not know the whole vector $\nabla_i \Phi(\xx{x}{t}{})$.
Instead, the only information player $i$ has after choosing a path $s_i$ is
$$c_{s_i}(\xx{X}{t}{}) = \sum_{e \in s_i} c_e(\ell_e(\xx{X}{t}{})), \;\mbox{ with } \ell_e(\xx{X}{t}{}) = \sum_{j\in N} \sum_{r\in \mc{S}_j:e\in r} \xx{X}{t}{j,r},$$
where $\xx{X}{t}{}$ is the choice vector of players at step $t$. Note that not only does each player $i$ receive less information (one value instead of $|\mc{S}_i|$ values), the information $c_{s_i}(\xx{X}{t}{})$ she receives actually is different from the more useful value $c_{s_i}(\xx{x}{t}{})$ that corresponds to an entry in $\nabla \Phi(\xx{x}{t}{})$.

In order to follow the framework in Section~\ref{sec:bulletin}, each player $i$ needs to have a way to approximate the vector $\nabla_i \Phi(\xx{x}{t}{})$, her portion of the gradient vector $\nabla \Phi(\xx{x}{t}{})$. Our basic idea is for each player to divide the time steps into episodes, each consisting of a consecutive number of steps, and to do the following in each episode. During episode $\tau$, each player $i$ plays some fixed strategy $\xx{x}{\tau}{i}$ for all the steps (instead of playing different strategies in different steps), collects statistics to obtain an estimate $\xx{\hat{g}}{\tau}{i}$ for $\nabla_i \Phi(\xx{x}{\tau}{})$, and at the end of the episode uses $\xx{\hat{g}}{\tau}{i}$ to update her strategy for the next episode. Two keys are: how to come up with the estimate $\xx{\hat{g}}{\tau}{i}$ and how to update the next strategy, which we describe next.

\subsubsection{Updating the Strategies} \label{sec:update}

With a ``good" estimate $\xx{\hat{g}}{\tau}{i}$ that we will define and show how to achieve in Section~\ref{sec:approx}, we can follow Section~\ref{sec:con} and consider the following update rule for each player $i$'s strategy of the next episode:
\begin{eqnarray}
\xx{x}{\tau+1}{i} &=& \arg\min_{z_i\in\bK_i} \left\{\eta_i \langle \xx{\hat g}{\tau}{i}, z_i\rangle + \BRi{z_i,\xx{x}{\tau}{i}}\right\} \label{eq:update1_}
\end{eqnarray}
Here, $\eta_i >0$ is some learning rate, $R_i:\bK_i \to \mathbb{R}$ is some regularization function, and $\BRi{\cdot, \cdot}$ is the Bregman divergence with respect to $R_i$ defined as
$$\BRi{u_i,v_i} = R_i(u_i) - R_i(v_i) - \langle \nabla R_i(v_i), u_i-v_i\rangle$$
for $u_i,v_i \in \bK_i$. This gives rise to a family of update rules for different choices of $R_i$. For example, it is known that by choosing $R_i(u_i) = \|u_i\|_2^2 /2$ to have $\BRi{x_i,y_i} = \|x_i - y_i\|_2^2 /2$, one recovers the gradient descent algorithm, while by choosing $R_i(u_i) = \sum_s (u_{i,s} \ln u_{i,s} -u_{i,s})$ to have $\BRi{x_i,y_i} = \sum_s x_{i,s} \ln (x_{i,s} / y_{i,s})$, one recovers the multiplicative updates algorithm. 
In the bandit atomic model, (with a properly chosen $R_i$) we may need the number of players $n$ to be large to have low regret\footnote{Later in Section~\ref{sec:conv2}, it will be clear that since our goal is to make the convergence close to the minimizer (considering the approximation error of the gradient), the learning rate bound $\eta$ there is set dependent on $n$ in a way such that $\eta$ is smaller when $n$ is larger.
This in turn affects that $n$ has to be large enough to make the dominant term, which is proportional to $\eta$, in the regret small enough.} for each player in order to provide an incentive for the players to use the algorithm.
We will need these $R_i$'s with $\BRi{\cdot,\cdot}$'s to satisfy the following, which is a bit stricter than Assumption~\ref{as:cond}.

\begin{assumption} \label{as:cond_}
There is some parameter $\Gamma$ such that for any $i \in N$ and any $x_i,y_i \in \bK_i$,
$$\Gamma \cdot \BRi{x_i,y_i} \le \|x_i-y_i\|_2^2 \le 2 \cdot \BRi{x_i,y_i}.$$
\end{assumption}
Note that the parameter $\Gamma$ in the assumption is determined by the choice of $R_i$ as well as the set $\bK_i$, which is a subset (defined in Section~\ref{sec:approx}) of $\mc{K}_i$ and in turn depends on the parameter $\Lambda$ introduced in Section~\ref{sec:approx}. It is clear that for $R_i(u_i) = \|u_i\|_2^2 /2$ and $\BRi{x_i,y_i} = \|x_i - y_i\|_2^2 /2$, the assumption holds with $\Gamma = 2$. For $R_i(x_i) = \sum_s (x_{i,s} \ln x_{i,s} -x_{i,s})$ and $\BRi{x_i,y_i} = \sum_s x_{i,s} \ln (x_{i,s} / y_{i,s})$, the following shows that the assumption holds with $\Gamma=\Lambda/n$, which we prove in Appendix~\ref{app:KL}.
\begin{proposition} \label{pro:KL}
With $\BRi{x_i,y_i} = \sum_s x_{i,s} \ln (x_{i,s} / y_{i,s})$, it holds that for any $x_i, y_i \in \bK_i$,
$$\frac{\Lambda}{n} \cdot \BRi{x_i,y_i} \le \|x_i - y_i\|_2^2 \le 2 \cdot \BRi{x_i,y_i}.$$
\end{proposition}

Let us remark that we use the $L_2$ norm $\|\cdot\|_2$ in the assumption, instead of a general norm, for the purpose of simplifying our presentation in the next section; one can check that our analysis there also works for a general norm, but the bounds derived would be more complicated.

\subsubsection{Approximating the Gradient} \label{sec:approx}
Now, we define and show how to guarantee a good estimate $\xx{\hat{g}}{\tau}{i}$.
Consider any episode $\tau$ and player $i$. For each path $s \in \mc{S}_i$, she computes the average cost of that path during the episode as the corresponding entry in $\xx{\hat{g}}{\tau}{i}$. That is, if we let $\xx{\hat{g}}{\tau}{i,s}$ denote the $s$'th entry in $\xx{\hat{g}}{\tau}{i}$ and let $\xx{\mc{T}}{\tau}{i,s}$ denote the set of steps in episode $\tau$ that player $i$ chose path $s$, then
\begin{equation}\label{eq:g}
\xx{\hat{g}}{\tau}{i,s} = \frac{1}{|\xx{\mc{T}}{\tau}{i,s}|} \sum_{t \in \xx{\mc{T}}{\tau}{i,s}} c_{s}(\xx{X}{t}{}),
\end{equation}
which we would like to be a good approximation of $c_s(\xx{x}{\tau}{})$. This requires $|\xx{\mc{T}}{\tau}{i,s}|$ to be large enough which in turns needs $\xx{\pi}{\tau}{i,s}$, or equivalently $\xx{x}{\tau}{i,s}$, to be large enough. To guarantee this, we restrict each player $i$ to play strategies in a smaller feasible set $\bK_i \subseteq \mc{K}_i$, such that for any $x_i \in \bK_i$, $x_{i,s} \ge \Lambda (1/n)$, for some parameter $\Lambda$, which implies that each path $s$ is chosen with probability $\pi_{i,s} \ge \Lambda$. One way of choosing such $\bK_i$ is for each $y_i \in \mc{K}_i$ to include a corresponding $x_i$ with $x_{i,s} = (1-\Lambda) y_{i,s} + \Lambda (1/n)$, for each $s$. In this way, one can have the guarantee that every $y_i \in \mc{K}_i$ has some $x_i \in \bK_i$ such that $\|x_i - y_i \|_1 \le d \Lambda (1/n)$. We sum up the procedure for approximating the gradient in the following.
\begin{algorithm}
\caption{Procedure for approximating the gradient}
\begin{algorithmic}[1]
\STATE for each $y_i \in \mc{K}_i$ include a corresponding $x_i \in \bK_i$ with $x_{i,s} = (1-\Lambda) y_{i,s} + \Lambda (1/n)$, for each $s$.
\FOR{any episode $\tau$ and player $i$}
\FOR{each path $s \in \mc{S}_i$}
\STATE Compute $\xx{\hat{g}}{\tau}{i,s} = \frac{1}{|\xx{\mc{T}}{\tau}{i,s}|} \sum_{t \in \xx{\mc{T}}{\tau}{i,s}} c_{s}(\xx{X}{t}{})$.
\ENDFOR
\ENDFOR
\end{algorithmic}
\end{algorithm}

Then by setting the number of steps in each episode $\tau$ to
$$\frac{\nu n^2 \log(nd \tau)}{\Lambda m^2},$$
for a large enough constant $\nu$, we have the following.

\begin{lemma} \label{lem:approx}
With high probability of at least $1-2\kappa$ for a small enough constant $\kappa$ chosen arbitrarily,
each player $i$ in each episode $\tau$ can have $$\|\xx{\hat{g}}{\tau}{i} - \nabla_i \Phi(\xx{x}{\tau}{})\|_\infty \le \frac{4 b m}{n}.$$
\end{lemma}
\begin{proof}
Fix any player $i$, path $s \in \mc{S}_i$, and episode $\tau$. Our goal is to show that
\begin{equation}\label{eq:diff}
\left|\xx{\hat{g}}{\tau}{i,s} - \nabla_{i,s} \Phi(\xx{x}{\tau}{})\right| = \left|\frac{1}{|\xx{\mc{T}}{\tau}{i,s}|} \sum_{t \in \xx{\mc{T}}{\tau}{i,s}} \left(c_{s}(\xx{X}{t}{}) - c_s(\xx{x}{\tau}{})\right) \right|
\end{equation}
is small with high probability. We would like to show that $c_{s}(\xx{X}{t}{})$ has expected value $\ex[c_{s}(\xx{X}{t}{})]$ close to $c_s(\xx{x}{\tau}{})$ so that we can apply a Hoeffding bound on (\ref{eq:diff}), but there is a subtlety here that we need to be careful about. Note that for a given $t$, although the expected value of $\xx{X}{t}{}$ equals $\xx{x}{\tau}{}$ when there is no additional conditioning, this is no longer true under the condition that $t \in \xx{\mc{T}}{\tau}{i,s}$, as player $i$'s choice does not remain random but is fixed to $s$ according to the definition of $\xx{\mc{T}}{\tau}{i,s}$.

Let us consider a related random variable $\xx{\hat{X}}{t}{}$, which is similar to $\xx{X}{t}{}$ but with player $i$'s choice left random, sampled according to $\xx{x}{\tau}{i}$. Then it has the nice property that $\ex[\xx{\hat{X}}{t}{}] = \xx{x}{\tau}{}$, so that
$\ex[\ell_e(\xx{\hat{X}}{t}{})] = \ell_e(\xx{x}{\tau}{})$
for any $e$ and hence
\begin{equation}\label{eq:c_s}
\left|\Ex{c_s(\xx{\hat{X}}{t}{})} - c_s(\xx{x}{\tau}{}) \right| \le \frac{Bm}{8n},
\end{equation}
for any $s$ by Proposition~\ref{pro:exp}. In addition, it has costs close to those associated with $\xx{X}{t}{}$, as
by Taylor's expansion on $c_e(\ell_e(\xx{X}{t}{}))$ at $\ell_e(\xx{\hat{X}}{t}{})$,
$| c_e(\ell_e(\xx{X}{t}{})) - c_e(\ell_e(\xx{\hat{X}}{t}{})) |$ is at most
$$b \left| \ell_e(\xx{X}{t}{}) - \ell_e(\xx{\hat{X}}{t}{}) \right| + \frac{b}{2} \left| \ell_e(\xx{X}{t}{}) - \ell_e(\xx{\hat{X}}{t}{}) \right|^2 \le \frac{b}{n} + \frac{b}{2 n^2} \le \frac{2b}{n}$$
for any $e$, which implies that
\begin{equation}\label{eq:S}
\left| c_s(\xx{X}{t}{}) - c_s(\xx{\hat{X}}{t}{}) \right| \le \sum_{e \in s} \frac{2b}{n} \le \frac{2bm}{n}
\end{equation}
for any $s$. From the bounds in (\ref{eq:S}) and (\ref{eq:c_s}), we have
\begin{eqnarray}
\lefteqn{\left|\Ex{c_{s}(\xx{X}{t}{})} - c_s(\xx{x}{\tau}{}) \right|}\nonumber\\
&\le& \left|\Ex{c_{s}(\xx{\hat{X}}{t}{})} - c_s(\xx{x}{\tau}{}) \right| + \left|\Ex{c_{s}(\xx{X}{t}{})} - \Ex{c_{s}(\xx{\hat{X}}{t}{})} \right|\nonumber\\
&\le& \frac{Bm}{8n} + \frac{2 bm}{n}\nonumber\\
&\le& \frac{3 b m}{n}. \label{eq:exp}
\end{eqnarray}

With this, we can then apply the Hoeffding bound to get
$$\pr{\left|\frac{1}{|\xx{\mc{T}}{\tau}{i,s}|} \sum_{t \in \xx{\mc{T}}{\tau}{i,s}} \left(c_{s}(\xx{X}{t}{}) - c_s(\xx{x}{\tau}{})\right) \right| > \frac{4 b m}{n}} \le \frac{\kappa}{n d \tau^2}$$
for some small enough constant $\kappa$, whenever $|\xx{\mc{T}}{\tau}{i,s}| \ge \xx{w}{\tau}{}$ where $\xx{w}{\tau}{} = \frac{\nu n^2 \log(nd \tau)}{2m^2}$ for a large enough constant $\nu$. Let us define the following good event
\begin{itemize}
  \item $G$: $|\xx{\mc{T}}{\tau}{i,s}| \ge \xx{w}{\tau}{}$ for every $\tau, i, s$.
\end{itemize}
Recall that each episode $\tau$ consists of $\frac{2\xx{w}{\tau}{}}{\Lambda}$ steps. Thus, for any $\tau, i, s$,
$$\Ex{|\xx{\mc{T}}{\tau}{i,s}|} \ge \Lambda \frac{2\xx{w}{\tau}{}}{\Lambda} = 2\xx{w}{\tau}{}$$
and a union bound together with a Hoeffding bound give us
$$\pr{\neg G} \le \sum_{\tau,i,s} \pr{|\xx{\mc{T}}{\tau}{i,s}| < \xx{w}{\tau}{}} \le \sum_{\tau,i,s} \frac{\kappa}{2 n d \tau^2} \le \kappa.$$
As a result, we have
\begin{eqnarray*}
\lefteqn{\pr{\exists \tau,i,s: \left|\xx{\hat{g}}{\tau}{i,s} - \nabla_{i,s} \Phi(\xx{x}{\tau}{})\right| > \frac{4 b m}{n}}}\\
&\le& \sum_{\tau,i,s} \pr{\left|\xx{\hat{g}}{\tau}{i,s} - \nabla_{i,s} \Phi(\xx{x}{\tau}{})\right| > \frac{4 b m}{n} \mid G} + \pr{\neg G}\\
&\le& \sum_{\tau,i,s} \frac{\kappa}{2 n d \tau^2} + \kappa\\
&\le& 2 \kappa,
\end{eqnarray*}
for a small enough constant $\kappa$, which proves the lemma.
\end{proof}

\subsection{Convergence Results} \label{sec:conv_}

In this section, we analyze the behavior of the system of players adopting the algorithm described in Section~\ref{sec:alg}.
Our main result is the following, which shows that the system indeed quickly converges, in the sense that the value of the potential function $\Phi(\xx{x}{\tau}{})$ quickly comes near the minimum value $\Phi(q)$, where $q=\arg\min_{z\in\mathcal{K}} \Phi(z)$, and then stays close afterwards.

\begin{theorem} \label{thm:Nash2}
Consider any atomic congestion game of $n$ players, with a potential function $\Phi$ which is $(\alpha,\beta,\lambda)$-smooth, and let $q = \arg\min_{z\in\mathcal{K}} \Phi(z)$. Suppose each player $i$ updates her strategy according to the rule in (\ref{eq:update1}), with $\eta_i \le 1/\lambda$, using $\xx{\hat{g}}{\tau}{i}$ described in (\ref{eq:g}) with the guarantee that
\begin{equation}\label{eq:err}
\left\|\xx{\hat{g}}{\tau}{i} - \nabla_i \Phi(\xx{x}{\tau}{}) \right\|_\infty \le \epsilon.
\end{equation}
Consider any $\eta$ such that $\eta \le \eta_i$ for any $i$ and $\theta = \sqrt{\eta \Gamma \epsilon n} \le 1$, and let $\delta = 6 \epsilon + \theta \beta d \Lambda$, where $\Gamma$ is the parameter in Assumption~\ref{as:cond_} and $\Lambda$ is the parameter introduced in Section~\ref{sec:approx}. Then for $\tau_0 = \alpha/\delta$ and $\triangle = 3\delta /\theta$, it holds that for any $\tau \ge \tau_0$,
$$\Phi(\xx{x}{\tau}{}) \le \Phi(q) + \triangle$$
or equivalently,
$$\Phi(\xx{x}{\tau}{}) \le \Phi(q) + 3(6\sqrt{\frac{\epsilon}{\eta \Gamma n}}+\beta d \Lambda).$$
\end{theorem}

As in Section~\ref{sec:con}, we consider two examples of setting the choices in the following, which will be proved in Section~\ref{sec:conv2}.

\begin{corollary} \label{cor:conv2}
Consider any atomic congestion game of $n$ players with parameters given in Section~\ref{sec:preliminaries}. If each player $i$ plays the gradient descent algorithm over $\bK_i$ with $\eta_i \le 1/\lambda$, then with high probability,
$$\triangle = \mc{O}\left(\frac{k^{1/2} m}{n}\right) \mbox{ and } T_0 = \mc{O}\left(\frac{n^4 d \log(nd)}{m^2 k^{1/2}}\right).$$
Furthermore, if each player $i$ plays the multiplicative updates algorithm over $\bK_i$ with $\eta_i \le 1/\lambda$, then with high probability,
$$\triangle = \mc{O}\left(\left(\frac{kd}{n}\right)^{1/3}m\right) \mbox{ and } T_0 = \mc{O}\left(\frac{n^{10/3} d^{2/3} \log(nd)}{m^2 k^{1/3}}\right).$$
\end{corollary}

From the corollary, we see that playing the gradient descent algorithm guarantees a smaller error $\triangle$, as $k \le d$. Take for example the load balancing game studied in \cite{kleinberg:piliouras:tardos:load}, which has $d=n$ and $k=m=1$. The gradient-descent algorithm can achieve $\triangle = \mc{O}(1/n)$, while the multiplicative updates algorithm only has $\triangle = \mc{O}(1)$. Although the convergence time $T_0$ of the gradient-descent algorithm may look higher in the the corollary, one can in fact make it smaller by choosing a smaller $\Lambda$ (with a slightly increased $\triangle$), so that both of its $T_0$ and $\triangle$ are smaller than those of the multiplicative updates algorithm. Note that to have a small error, both algorithms need $n$, the number of players, to be large. This seems unavoidable in the bandit model considered here, because the amount of flow each player has is as large as $1/n$, which limits the accuracy each player can have on the estimations of the path costs and gradient vectors.

Implication of the convergence, including approximate equilibria and guarantees regarding social costs, will be discussed in Section~\ref{sec:imp}.

\subsubsection*{Analysis}
The proof of the theorem is based on the idea of Section~\ref{sec:con} that if each player $i$ can actually have her portion $\nabla_i\Phi(\xx{x}{\tau}{})$ of the gradient vector $\nabla\Phi(\xx{x}{\tau}{})$ and do the update
$$\xx{x}{\tau+1}{i} = \arg\min_{z_i\in\mc{K}_i} \left\{\eta_i \langle \nabla_i \Phi(\xx{x}{\tau}{}), z_i\rangle + \BRi{z_i,\xx{x}{\tau}{i}}\right\},$$
then the collective update of all players can be seen as doing some generalized mirror descent on the convex function $\Phi$. Then it was shown in Section~\ref{sec:con} that doing such a mirror descent on any convex function will lead to a quick convergence. However, in our setting, each player can only obtain an estimate $\xx{\hat{g}}{\tau}{i}$ of the desired vector $\nabla_i\Phi(\xx{x}{\tau}{})$, 
which only allows players collectively to perform an ``approximate" mirror descent on the convex function.
Unfortunately, the convergence analysis of Section~\ref{sec:con} relies on showing that $\Phi(\xx{x}{\tau+1}{}) \le \Phi(\xx{x}{\tau}{})$ for every $\tau$, which depends crucially on being able to move (in the mirror space) precisely in the opposite direction of the gradient vector. As we now can no longer move in that precise direction, the next $\Phi$ value may increase, and it is not clear if convergence can still be guaranteed. Interestingly, we provide a positive answer, with the following theorem. In fact it works for any general convex function $\Phi$, which is smooth in the sense of Definition~\ref{def:smooth}.

\begin{theorem} \label{thm:GD2}
Let $\Phi: \mc{K}\rightarrow\mathbb{R}$ be any $(\alpha,\beta,\lambda)$-smooth convex function, with $q = \arg\min_{z\in\mathcal{K}} \Phi(z)$, where $\mc{K}=\mc{K}_1 \times \cdots \times \mc{K}_n$, such that $\|x - y\|_1 \le 2$ for any $x,y \in \mc{K}$. Suppose for each $i$, we have a convex $\bK_i \subseteq \mc{K}_i$ with the property that any $x \in \mc{K}$ has some $y \in \bK=\bK_1 \times \cdots \times \bK_n$ such that $\|x-y\|_1 \le d \Lambda$,
and suppose moreover that we have $R_i$'s satisfying Assumption~\ref{as:cond_} so $\Gamma$ is the parameter therein.
Now consider the update rule in (\ref{eq:update1}), with each $\eta_i \le 1/\lambda$, using $\xx{\hat{g}}{\tau}{i}$ such that
$$\|\xx{\hat{g}}{\tau}{i} - \nabla \Phi_i(\xx{x}{\tau}{}) \|_\infty \le \epsilon.$$
Consider any $\eta$ such that $\eta \le \eta_i$ for any $i$ and $\theta = \sqrt{\eta \Gamma \epsilon n} \le 1$, and let $\delta = 6 \epsilon + \theta \beta d \Lambda$. Then for $\tau_0 = \alpha/\delta$ and $\triangle = 3\delta /\theta$, it holds that for any $\tau \ge \tau_0$,
$$\Phi(\xx{x}{\tau}{}) \le \Phi(q) + \triangle$$ or equivalently,
$$\Phi(\xx{x}{\tau}{}) \le \Phi(q) + 3(6\sqrt{\frac{\epsilon}{\eta \Gamma n}}+\beta d \Lambda).$$
\end{theorem}


Note that Theorem~\ref{thm:Nash2} follows immediately from Theorem~\ref{thm:GD2}. This is because our function $\Phi$ is convex and furthermore, any $x \in \mc{K}$ has some $y \in \bK$ with $\|x - y \|_1 = \sum_i \|x_i - y_i \|_1 \le  \sum_i d\Lambda (1/n) = d \Lambda$. On the other hand, Theorem~\ref{thm:GD} works for a general convex function, which may have independent interest of its own. To prove Theorem~\ref{thm:GD2}, we rely on the following key lemma.

\begin{lemma} \label{lem:less2}
Let $\theta = \sqrt{\eta \Gamma \epsilon n}$ and $\delta = 6 \epsilon + \theta \beta d \Lambda$. Then given the assumption of Theorem~\ref{thm:GD2}, we have
$$\Phi(\xx{x}{\tau+1}{}) \le \Phi(\xx{x}{\tau}{}) - \theta \left(\Phi(\xx{x}{\tau}{}) - \Phi(q)\right) + \delta,$$
for any integer $\tau \ge 0$. This means that
$$\Phi(\xx{x}{\tau+1}{}) \le \Phi(\xx{x}{\tau}{}) - \delta$$
whenever $\Phi(\xx{x}{\tau}{}) - \Phi(q) \ge 2\delta / \theta$.
\end{lemma}

This lemma helps us deal with the difficulty we discussed before. That is, although doing mirror descent using the exact gradient vector can guarantee a decreased $\Phi$ value, as shown in Section~\ref{sec:con}, this may no longer hold in general when using an approximate gradient vector. Nevertheless, Lemma~\ref{lem:less2} shows that as long as the current $\Phi$ value has a large enough gap from the minimum one, doing mirror descent using an approximate gradient vector can still guarantee a decreased $\Phi$ value.
We will prove the lemma in Section~\ref{sec:less2}. Assuming the lemma, now we proceed to prove Theorem~\ref{thm:GD2}.

\begin{proof} [Proof of Theorem~\ref{thm:GD2}]
From Lemma~\ref{lem:less2}, we know that as long as $\Phi(\xx{x}{\tau}{}) - \Phi(\xx{q}{}{}) \ge 2\delta / \theta$, the next $\Phi(\xx{x}{\tau+1}{})$ will decrease from $\Phi(\xx{x}{\tau}{})$ by at least the amount $\delta$. Since $\Phi(x) \le \alpha$ for any $x \in \mc{K}$, according to the smoothness assumption, and $\Phi(q) \ge 0$, we know that there must be some episode $\tau_1 \le \tau_0 = \alpha/\delta$
such that $\Phi(\xx{x}{\tau_1}{}) - \Phi(q) < 2\delta / \theta$. After that, the value of $\Phi$ may again exceed the threshold $\Phi(\xx{q}{}{}) + 2\delta / \theta$, but we next argue that it cannot go much higher than that.

Let $\tau_2$ denote the episode after $\tau_1$ such that $\Phi(\xx{x}{\tau_2}{})$ reaches the highest value. Then $\Phi(\xx{x}{\tau_2}{})$ much have just been increased from its previous episode with a smaller $\Phi$ value, which can only happen when
$$\Phi(\xx{x}{\tau_2-1}{}) - \Phi(\xx{q}{}{}) \le 2\delta / \theta.$$
On the other hand, we know from Lemma~\ref{lem:less2} that
$$\Phi(\xx{x}{\tau_2}{}) \le \Phi(\xx{x}{\tau_2-1}{}) + \delta.$$
Consequently, for any $\tau \ge \tau_0 \ge \tau_1$, we must have
$$\Phi(\xx{x}{\tau}{}) \le \Phi(\xx{x}{\tau_2}{}) \le \Phi(\xx{q}{}{}) + 2\delta / \theta + \delta \le \Phi(\xx{q}{}{}) + 3 \delta / \theta,$$
which completes the proof of Theorem~\ref{thm:GD2}.
\end{proof}

\subsubsection{Proof of Lemma~\ref{lem:less2}} \label{sec:less2}

To simplify our notation, let us denote the gradient vector $\nabla \Phi(\xx{x}{\tau}{})$ by $\xx{g}{\tau}{}=(\xx{g}{\tau}{1}, \dots, \xx{g}{\tau}{n})$, with $\xx{g}{\tau}{i}=\nabla_i \Phi(\xx{x}{\tau}{})$.
Then from the assumption that $\Phi$ is $(\alpha,\beta,\lambda)$-smooth, we have
$$\Phi(\xx{x}{\tau+1}{}) \le \Phi(\xx{x}{\tau}{}) + \langle \xx{g}{\tau}{}, \xx{x}{\tau+1}{}-\xx{x}{\tau}{}\rangle + \frac{\lambda}{2} \|\xx{x}{\tau+1}{}-\xx{x}{\tau}{}\|_2^2,$$
with the right-hand side above being
\begin{eqnarray*}
\lefteqn{\Phi(\xx{x}{\tau}{}) + \sum_{i \in N} \left(\langle \xx{g}{\tau}{i}, \xx{x}{\tau+1}{i}-\xx{x}{\tau}{i}\rangle + \frac{\lambda}{2} \|\xx{x}{\tau+1}{i}-\xx{x}{\tau}{i}\|_2^2\right)} \\
&\le& \Phi(\xx{x}{\tau}{}) + \sum_{i \in N} \left(\langle \xx{g}{\tau}{i}, \xx{x}{\tau+1}{i}-\xx{x}{\tau}{i}\rangle + \frac{1}{\eta_i} \BRi{\xx{x}{\tau+1}{i},\xx{x}{\tau}{i}}\right), 
\end{eqnarray*}
according to Assumption~\ref{as:cond_} together with the assumption that $\eta_i \le 1/\lambda$ and hence $\lambda \le 1/\eta_i$, as well as the fact that $\BRi{\xx{x}{\tau+1}{},\xx{x}{\tau}{}}$ is nonnegative. To bound the sum above, let us first change $\langle \xx{g}{\tau}{i}, \xx{x}{\tau+1}{i}-\xx{x}{\tau}{i}\rangle$ into
$$\langle \xx{\hat{g}}{\tau}{i}, \xx{x}{\tau+1}{i}-\xx{x}{\tau}{i}\rangle + \langle \xx{g}{\tau}{i} - \xx{\hat{g}}{\tau}{i}, \xx{x}{\tau+1}{i}-\xx{x}{\tau}{i}\rangle,$$
so that each
$$\langle \xx{g}{\tau}{i}, \xx{x}{\tau+1}{i}-\xx{x}{\tau}{i}\rangle + \frac{1}{\eta_i} \BRi{\xx{x}{\tau+1}{i},\xx{x}{\tau}{i}}$$
now becomes
$$\langle \xx{\hat{g}}{\tau}{i}, \xx{x}{\tau+1}{i}-\xx{x}{\tau}{i}\rangle + \frac{1}{\eta_i} \BRi{\xx{x}{\tau+1}{i},\xx{x}{\tau}{i}} + \langle \xx{g}{\tau}{i}-\xx{\hat{g}}{\tau}{i}, \xx{x}{\tau+1}{i}-\xx{x}{\tau}{i}\rangle.$$
The sum of the first two terms above according to the definition of $\xx{x}{\tau+1}{i}$ in (\ref{eq:update1}) is at most
$$\langle \xx{\hat{g}}{\tau}{i}, z_i-\xx{x}{\tau}{i}\rangle + \frac{1}{\eta_i} \BRi{z_i,\xx{x}{\tau}{i}},$$
for any $z_i \in \bK_i$, while the last term above by the Cauchy-Schwarz inequality is at most
$$\| \xx{g}{\tau}{i}-\xx{\hat{g}}{\tau}{i}\|_\infty \|\xx{x}{\tau+1}{i}-\xx{x}{\tau}{i}\|_1 \le \epsilon \|\xx{x}{\tau+1}{i}-\xx{x}{\tau}{i}\|_1,$$
according to the assumption that $\| \xx{g}{\tau}{i}-\xx{\hat{g}}{\tau}{i}\|_\infty \le \epsilon$. Since
$$\sum_{i \in N} \epsilon\|\xx{x}{\tau+1}{i}-\xx{x}{\tau}{i}\|_1 = \epsilon\|\xx{x}{\tau+1}{}-\xx{x}{\tau}{}\|_1 \le 2 \epsilon$$
from the assumption of the lemma, we have
\begin{eqnarray}
\lefteqn{\sum_{i \in N} \left(\langle \xx{g}{\tau}{i}, \xx{x}{\tau+1}{i}-\xx{x}{\tau}{i}\rangle + \frac{1}{\eta_i} \BRi{\xx{x}{\tau+1}{i},\xx{x}{\tau}{i}}\right)} \nonumber\\
&\le& \sum_{i \in N} \langle \xx{\hat{g}}{\tau}{i}, z_i-\xx{x}{\tau}{i}\rangle + \frac{1}{\eta} \sum_{i \in N} \BRi{z_i,\xx{x}{\tau}{i}} + 2\epsilon, \label{eq:sum}
\end{eqnarray}
for any $z_i \in \bK_i$.

Our goal then is to choose these $z_i$'s to make the above small enough. Our idea is to take $z=(z_1, \dots, z_n)$ as a point that moves from $\xx{x}{\tau}{}$ towards
$$\bq = (\bq_1, \dots, \bq_n) = \arg\min_{x \in \bK} \Phi(x).$$
That is, we consider $$z_i=\xx{x}{\tau}{i} + \theta (\bq_i - \xx{x}{\tau}{i}),$$
for some parameter $\theta \in [0,1]$ to be determined later. Note that $z_i$ does belong to $\bK_i$ because it is a convex combination of $\xx{x}{\tau}{i}$ and $\bq_i$, both of which are in the convex set $\bK_i$. With $z_i-\xx{x}{\tau}{i} = \theta (\bq_i - \xx{x}{\tau}{i})$, we are now able to relate the bound in (\ref{eq:sum}) to $\Phi(\xx{x}{\tau}{}) - \Phi(\bq)$ because the first term there is
\begin{eqnarray}
\lefteqn{\sum_{i \in N} \langle \xx{\hat{g}}{\tau}{i}, z_i-\xx{x}{\tau}{i}\rangle} \nonumber\\
&=& \sum_{i \in N} \langle \xx{g}{\tau}{i}, z_i-\xx{x}{\tau}{i}\rangle + \sum_{i \in N} \langle \xx{\hat{g}}{\tau}{i}-\xx{g}{\tau}{i}, z_i-\xx{x}{\tau}{i}\rangle \nonumber\\
&=& \theta \sum_{i \in N} \langle \xx{g}{\tau}{i}, \bq_i - \xx{x}{\tau}{i} \rangle + \theta \sum_{i \in N} \langle \xx{\hat{g}}{\tau}{i}-\xx{g}{\tau}{i}, \bq_i - \xx{x}{\tau}{i} \rangle, \label{eq:z-x}
\end{eqnarray}
where the first term in (\ref{eq:z-x}) is
$$\theta \langle \xx{g}{\tau}{}, \bq - \xx{x}{\tau}{} \rangle \le - \theta \left(\Phi(\xx{x}{\tau}{}) - \Phi(\bq)\right)$$
by the convexity of $\Phi$, while the second term in (\ref{eq:z-x}) by the Cauchy-Schwarz inequality is at most
$$\theta \sum_{i \in N} \|\xx{\hat{g}}{\tau}{i}-\xx{g}{\tau}{i}\|_\infty \|z_i-\xx{x}{\tau}{i}\|_1 \le \theta \epsilon \|z -\xx{x}{\tau}{}\|_1 \le 2\epsilon,$$
as $\|z -\xx{x}{\tau}{}\|_1 \le 2$ and $\theta \le 1$. It remains to bound the second term in (\ref{eq:sum}). Note that according to Assumption~\ref{as:cond_},
$$\BRi{z_i,\xx{x}{\tau}{i}} \le \frac{1}{\Gamma} \|z_i-\xx{x}{\tau}{i}\|_2^2 = \frac{\theta^2}{\Gamma} \|\bq_i - \xx{x}{\tau}{i}\|_2^2,$$
where
$$\|\bq_i - \xx{x}{\tau}{i}\|_2^2 = \sum_s \left(\bq_{i,s} - \xx{x}{\tau}{i,s}\right)^2 \le \sum_s \frac{1}{n}\left|\bq_{i,s} - \xx{x}{\tau}{i,s}\right|,$$
since $\bq_{i,s},\xx{x}{\tau}{i,s} \in [0,\frac{1}{n}]$ for every $s$. This means that
$$\frac{1}{\eta} \sum_{i \in N} \BRi{z_i,\xx{x}{\tau}{i}} \le \frac{\theta^2}{\eta \Gamma n} \|\xx{x}{\tau}{} - \bq\|_1 \le \frac{2\theta^2}{\eta \Gamma n}.$$
By plugging these bounds into (\ref{eq:sum}), we have
\begin{eqnarray*}
\lefteqn{\sum_{i \in N} \left(\langle \xx{g}{\tau}{i}, \xx{x}{\tau+1}{i}-\xx{x}{\tau}{i}\rangle + \frac{1}{\eta_i} \BRi{\xx{x}{\tau+1}{i},\xx{x}{\tau}{i}}\right)}\\
&\le& - \theta \left(\Phi(\xx{x}{\tau}{}) - \Phi(\bq)\right) + 2\epsilon + \frac{2\theta^2}{\eta \Gamma n} + 2 \epsilon,
\end{eqnarray*}
which then implies that
\begin{equation}\label{eq:bq}
\Phi(\xx{x}{\tau+1}{}) \le \Phi(\xx{x}{\tau}{}) - \theta \left(\Phi(\xx{x}{\tau}{}) - \Phi(\bq)\right) + \frac{2\theta^2}{\eta \Gamma n} + 4 \epsilon,
\end{equation}
for any $\theta \in [0,1]$.

The bound in (\ref{eq:bq}) is still not satisfactory as it involves $\bq$ instead of $q$. To relate $\Phi(\bq)$ to $\Phi(q)$, note that according to the assumption, there exists some $q' \in \bK$ such that $\|q-q' \|_1 \le d \Lambda,$ while as $\bq$ minimizes $\Phi$ over $\bK$, we have
$$\Phi(\bq) \le \Phi(q') \le \Phi(q) + \langle \nabla \Phi(q'), q'-q \rangle$$
by the convexity of $\Phi$. Then, by the Cauchy-Schwarz inequality, we have
$$\Phi(\bq) \le \Phi(q) + \|\nabla \Phi(q')\|_\infty \|q'-q \|_1 \le \Phi(q) + \beta d \Lambda,$$
as $\|\nabla \Phi(q')\|_\infty \le\beta$ by the smoothness assumption on $\Phi$. Finally, by substituting this into (\ref{eq:bq}), we have
$$\Phi(\xx{x}{\tau+1}{}) \le \Phi(\xx{x}{\tau}{}) - \theta \left(\Phi(\xx{x}{\tau}{}) - \Phi(q)\right) + \frac{2\theta^2}{\eta \Gamma n} + 4 \epsilon + \theta \beta d \Lambda.$$

Now let us choose
$\theta = \sqrt{\eta \Gamma n\epsilon}$, and note that $\theta \le 1$ by the assumption on $\eta$. With this choice of $\theta$, we obtain
$$\Phi(\xx{x}{\tau+1}{}) \le \Phi(\xx{x}{\tau}{}) - \theta \left(\Phi(\xx{x}{\tau}{}) - \Phi(q)\right) + \delta,$$
where $$\delta = 6 \epsilon + \theta \beta d \Lambda.$$
This implies that $\Phi(\xx{x}{\tau+1}{}) \le \Phi(\xx{x}{\tau}{}) - \delta$
whenever $\Phi(\xx{x}{\tau}{}) - \Phi(q) \ge 2\delta/\theta$,
which completes the proof of Lemma~\ref{lem:less2}.

\subsubsection{Proof of Corollary~\ref{cor:conv2}} \label{sec:conv2}
Note that the bounds in the theorem depend on the parameters $\Gamma$ and $\Lambda$. In fact, we have the freedom to choose the parameter $\Lambda$ as well as the functions $R_i(\cdot)$'s, while the parameter $\Gamma$ is then determined by them.

The first is to choose $R_i(x_i) = \|x_i\|_2^2 /2$ to have $\BRi{x_i,y_i} = \|x_i - y_i\|_2^2 /2$. In this case, each player's algorithm becomes the gradient-descent algorithm, and we can choose
$$\Gamma = 2\mbox{ and }  \Lambda = \sqrt{\frac{\epsilon}{2\eta n}} \frac{1}{\beta d},$$
which gives us
$$\triangle = 21 \sqrt{\frac{\epsilon}{2\eta n}} \mbox{ and } \tau_0 = \frac{\alpha}{7 \epsilon}.$$
The second example is to choose $R_i(x_i) = \sum_s (x_{i,s} \ln x_{i,s} -x_{i,s})$ to have $\BRi{x_i,y_i} = \sum_s x_{i,s} \ln (x_{i,s} / y_{i,s})$. In this case, each player's algorithm becomes the multiplicative updates algorithm, and according to Proposition~\ref{pro:KL} we can choose
$$\Gamma = \frac{\Lambda}{n} \mbox{ and } \Lambda = \left(\frac{\epsilon}{\eta \beta^2 d^2}\right)^{1/3},$$
which gives us
$$\triangle = 21 \left(\frac{\epsilon \beta d}{\eta} \right)^{1/3} \mbox{ and } \tau_0 = \frac{\alpha}{7 \epsilon}.$$
Although we stated above the convergence bounds in terms of the number of episodes, by incorporating the sizes of the episodes, it follows that the number of steps needed for convergence is at most
$$T_0 = \frac{\nu n^2 \log(nd \tau_0)}{\Lambda m^2} \cdot \tau_0.$$

Finally, let us go back to the congestion game and substitute the parameters $\alpha=bm/2, \beta=bm, \lambda=bmk$ from Proposition~\ref{pro:para} into the bounds above. Let us set the parameter $\epsilon=4bm/n$ according to Lemma~\ref{lem:approx} so that (\ref{eq:err}) holds for every player and episode with high probability. Moreover, let us assume for simplicity that $\eta_i \ge \Omega(1/\lambda)$ so that we have $\eta \ge \Omega(1/\lambda)$. Then we have the corollary.

\subsection{Equilibria and Social Costs} \label{sec:imp}

Now we briefly discuss the implication of the guarantee $\Phi(\xx{x}{\tau}{}) \le \Phi(q)+\triangle$
given by the results in Section~\ref{sec:conv_}.
The first is that such an $\xx{\pi}{\tau}{}$ is an approximate equilibrium in mixed strategies,
where we say that $\pi$ is a $\delta$-equilibrium if for any player $i \in N$ and any paths $s, s' \in \mc{S}_i$ with $x_{i,s}>0$,
$\ex[c_s(X)] \le \ex[c_{s'}(X)] + \delta$, where $X$ is the choice vector sampled according to $x$ and $x_i=\frac{1}{n}\pi_i$ for any $i$.\footnote{There is an alternative way to define an approximate equilibrium.
Let $(y_i(t),X_{-i})$ denote a vector where with no randomness $y_i(t)$ is a vector with $\frac{1}{n}$ in dimension~$t$ and all $0$'s in the rest and still with randomness $X_{-i}$ is the choice vector sampled according to $x$ without dimension $i$ and $x_i=\frac{1}{n}\pi_i$ for any $i$.
If $\pi$ is a $\delta$-equilibrium if for any player $i \in N$ and any paths $s, s' \in \mc{S}_i$ with $x_{i,s}>0$,
$\ex[c_s(y_i(s),X_{-i})] \le \ex[c_{s'}(y_i(s'),X_{-i})] + \delta$, where $(y_i(s),X_{-i}),(y_i(s'),X_{-i})$ are defined above.
We then have that $\ex[c_s(y_i(s),\xx{X}{\tau}{-i})] \le \ex[c_{s'}(y_i(s'),\xx{X}{\tau}{-i})] + \sqrt{8 b m\triangle} + \frac{3bm}{n}$ since any $\xx{x}{\tau}{}$ satisfying the condition $\Phi(\xx{x}{\tau}{}) \le \Phi(q)+\triangle$ must have $c_s(\xx{x}{\tau}{}) \le c_{s'}(\xx{x}{\tau}{}) + \sqrt{8 b m\triangle}$ along with Inequality~(\ref{eq:exp}).
Note that $\frac{bm}{n}$ can be large in general to make such equilibrium less meaningful due to bad approximation,
yet the difference bound is much smaller for some classes of structures of allowed paths.
For example, in load-balancing games, $\ex[c_s(\xx{X}{\tau}{})] \le \ex[c_{s'}(\xx{X}{\tau}{})] + \sqrt{8 b m\triangle}+\frac{3b}{n}$ since any $s$ consists of only one edge,
instead of at most $m$ edges.}
To show this, note that we know from Section~\ref{sec:approximate} that any  $\xx{x}{\tau}{}$ satisfying the condition $\Phi(\xx{x}{\tau}{}) \le \Phi(q)+\triangle$ must have $c_s(\xx{x}{\tau}{}) \le c_{s'}(\xx{x}{\tau}{}) + \sqrt{8 b m\triangle}$.
This and Proposition~\ref{pro:exp} together imply that $\ex[c_s(\xx{X}{\tau}{})]-\frac{Bm}{n} \le \ex[c_{s'}(\xx{X}{\tau}{})] + \sqrt{8 b m\triangle} \le \ex[c_{s'}(\xx{X}{\tau}{})] + \sqrt{8 b m\triangle} $, which gives that $\ex[c_s(\xx{X}{\tau}{})] \le \ex[c_{s'}(\xx{X}{\tau}{})] + \sqrt{8 b m\triangle} + \frac{Bm}{n}$, where $\xx{X}{\tau}{}$ is the choice vector sampled according to $\xx{x}{\tau}{}$.

$\frac{Bm}{n}$ can be large in general to make such equilibrium meaningless due to bad approximation.
Nevertheless, there are natural cases when such approximate equilibrium is meaningful.
As mentioned in Proposition~\ref{pro:exp}, for classes of cost functions with small $B$, small $\frac{Bm}{n}$ can be obtained to give meaningful $\delta$-equilibria.
For example, in particular for linear cost functions where $B=0$, $\ex[c_s(\xx{X}{\tau}{})] \le \ex[c_{s'}(\xx{X}{\tau}{})] + \sqrt{8 b m\triangle}$ since $\ex[c_{s}(\xx{X}{\tau}{})]=c_s(\xx{x}{\tau}{})$ for any $s$;
more generally, even if $B$ is not 0, it could be inherently small enough for some classes of cost functions to make $\frac{Bm}{n}$ small,
recalling that $B$ (defined in Section~\ref{sec:preliminaries}) is a constant with respect to the load, not necessarily to $\frac{m}{n}$.
For some classes of structures of allowed paths, the difference bound is smaller.
For example, in load-balancing games of \cite{kleinberg:piliouras:tardos:load}, it becomes that $\ex[c_s(\xx{X}{\tau}{})] \le \ex[c_{s'}(\xx{X}{\tau}{})] + \sqrt{8 b m\triangle}+\frac{B}{n}$ since any $s$ consists of only one edge,
instead of at most $m$ edges.
Using similar analyses to those in Section~\ref{sec:splittable}, we can also derive bounds on the ratio of a social cost at convergence to that at optimum.

\section{Conclusions and Future Work} \label{sec:conclusions}

We show that the mirror-descent dynamics converges to an approximate equilibrium in nonatomic congestion games. We do this by observing that the dynamics corresponds to a mirror-descent process on a convex potential function of such a game and then proving that the process converges to the minimum of the function. Moreover, we provide bounds on the outcome quality achieved by our dynamics in terms of two social costs: the average individual cost and the maximum individual cost. Finally, we propose a new family of bandit algorithms and show that when each player adopt such an algorithm in an atomic congestion game, their actual joint strategy profile quickly approaches an approximate Nash equilibrium.


There may be other no-regret or even other learning algorithms which could guarantee nice convergence properties or simply good quality of outcomes.
There are more learning algorithms and dynamics to be explored in repeated games, while classes of games are even more numerous.
Beyond learning, there is still a variety of different dynamics to consider in repeated games.

A different line of future work would be to consider appropriate bandit scenarios for market equilibrium problems and to see if generalized mirror-descents with approximate gradients also work there.
Yet another line of work could be extending our framework to other partial-information models by other suitable gradient estimation methods.


\bibliographystyle{plain}


\appendix

\section{Proof of Proposition~\ref{pro:conv}} \label{app:conv}

Recall that
$$\Phi(x)= \sum_{e\in E} \int_0^{\ell_e(x)}c_e(y)dy,$$
where $\ell_e(x) =\sum_{i\in N}\sum_{s:e\in s}x_{i,s}$. Let
$$\psi_e(v)=\int_0^{v}c_e(y)dy$$
so that $\Phi(x)=\sum_{e\in E} \psi_e(\ell_e(x))$. Observe that $\ell_e$ is a linear function of $x \in \mc{K}$, while $\psi_e$ is a convex function of $v \in \mathbb{R}$ as $c_e$ is assumed to be nondecreasing. Then for any $\delta \in [0,1]$ and any $x,x'\in \mc{K}$,
\begin{eqnarray*}
\lefteqn{(1-\delta)\Phi(x)+\delta\Phi(x')}\\
&=& \sum_{e\in E}\left((1-\delta) \psi_e(\ell_e(x))+ \delta \psi_e(\ell_e(x'))\right) \\
&\ge& \sum_{e\in E}\psi_e((1-\delta)\ell_e(x)+\delta \ell_e(x')) \\
&=& \sum_{e\in E}\psi_e(\ell_e((1-\delta)x+\delta x')) \\
&=& \Phi((1-\delta)x+\delta x').
\end{eqnarray*}
This proves that $\Phi$ is convex.


%

\section{Proof of Proposition~\ref{pro:para}} \label{app:para2}
Consider any $x \in \mc{K}$. First, to bound $\Phi(x)$, recall that $c_e(y) \le b y$ for any $y$, so that we have
\begin{eqnarray*}
\Phi(x) &=& \sum_{e \in E} \int_0^{\ell_e(x)} c_e(y) dy\\
&\le& \sum_{e \in E} \int_0^{\ell_e(x)} b y dy\\
&=&  \frac{b}{2} \sum_{e \in E} (\ell_e(x))^2.
\end{eqnarray*}
Since each player has a flow of amount $\frac{1}{n}$, we know that $\ell_e(x) \le 1$ and hence
\begin{eqnarray*}
\sum_{e \in E} (\ell_e(x))^2 &\le& \sum_{e \in E} \ell_e(x)\\
&=& \sum_{e \in E} \sum_{j \in N} \sum_{r \in \mc{S}_j: e \in r} x_{j,r}\\
&=& \sum_{j \in N} \sum_{r \in \mc{S}_j} \sum_{e \in r} x_{j,r}\\
&=& m,
\end{eqnarray*}
which implies that
$\Phi(x) \le \frac{bm}{2}.$

Next, to bound $\|\nabla \Phi(x)\|_\infty$, note that for any $i \in N$ and $s \in \mc{S}_i$, the entry in $\nabla \Phi(x)$ indexed by $(i,s)$ is
$$c_s(x) = \sum_{e \in s} c_e(x) \le bm.$$
Therefore, we have
$\|\nabla \Phi(x)\|_\infty \le b m.$

Finally, let us bound $\nabla^2 \Phi(x)$. Consider any $i,j \in N$, $s \in \mc{S}_i$, and $r \in \mc{S}_j$.
We know that
$$\frac{\partial^2\Phi(x)}{\partial x_{i,s} \partial x_{j,r}} = \sum_{e\in s\cap r} c'_e(\ell_e(x)).$$
From this, we know that each entry of the Hessian matrix $\nabla^2 \Phi(x)$
is either zero or at most $b m$. Consider any $z \in \mathbb{R}^d$. For any $i \in N$ and $s \in \mc{S}_i$, let $\Gamma_i(s)$ denote the set of paths in $\mc{S}_i$ that intersects $s$. Then
\begin{eqnarray*}
z^\top (\nabla^2 \Phi(x)) z &\le& \sum_{(i,s),(j,r)} \frac{\partial^2\Phi(x)}{\partial x_{i,s} \partial x_{j,r}} |z_{i,s}| |z_{i,r}|\\
&\le& b m \sum_i \sum_s |z_{i,s}| \sum_{r \in \Gamma_i(s)} |z_{i,r}|\\
&\le& b m \sum_i \sum_s |z_{i,s}| \sqrt{k \sum_{r \in \Gamma_i(s)} |z_{i,r}|^2}\\
&\le& b m \sum_i \sqrt{\sum_s |z_{i,s}|^2} \sqrt{\sum_s k \sum_{r \in \Gamma_i(s)} |z_{i,r}|^2},
\end{eqnarray*}
where the last two inequalities both use the Cauchy-Schwarz inequality. Now observe that
$$\sum_s \sum_{r \in \Gamma_i(s)} |z_{i,r}|^2 = \sum_r \sum_{s \in \Gamma_i(r)} |z_{i,r}|^2 \le k \sum_r |z_{i,r}|^2$$
and therefore we have
$$z^\top (\nabla^2 \Phi(x)) z \le b m \sum_i k \sum_s |z_{i,s}|^2 = b m k \|z\|_2^2$$
for any $z$. This implies that $\nabla^2 \Phi(x) \preceq \lambda I$ with $\lambda = b m k$, which proves the proposition.

\section{The No-Regret Property of Our Dynamics} \label{app:no-regret}
Mirror descents is known to have the no-regret property even
when the cost functions of the edges can vary with
time. We consider such a setting, where the actual cost of each edge at time step $t$ is $c^t_e(\ell_e(x^t))$,
where $\ell_e(x^t)$ is the load of the edge in question at $t$.
We will define a new cost function $C^t_e(x)$ as follows:
\[C^t_e(x):=\left\{\begin{array}{ll}
c^t_e(x)&\mbox{if $x\leq\ell_e(x^t)$}\\
c^t_e(\ell_e(x^t))&\mbox{otherwise.}
\end{array}\right.\]
Under these new cost functions, the cost of any edge observed at time step $t$, $c^t_e(\ell_e(x^t))$,
is actually the worst possible and any further increase on the load of any edge would have no effect on its cost.
If this optimistic view of the costs were actually true, then the algorithm using bulletin-board posting would perform
exactly as the mirror-descent algorithm and thus preserve the no-regret property.
Yet, the actual cost of any strategy under the real cost
functions $c$, when taking into account the effect of the infinitesimal deviating
agent (of a player), would be at least as bad as that under the optimistic
costs $C$. Thus, the performance of our algorithm is also
of $\epsilon$-regret for $\epsilon\rightarrow 0$ in regards to the best strategy in hindsight under the true costs.


\section{Proof of Proposition~\ref{pro:KL}} \label{app:KL}

According to the definition,
\begin{eqnarray*}
\BRi{x_i,y_i} &=& \sum_s x_{i,s} \ln \left(1+ \frac{x_{i,s} - y_{i,s}}{ y_{i,s}}\right)\\
&\le& \sum_s x_{i,s} \frac{x_{i,s} - y_{i,s}}{y_{i,s}}\\
&=& \sum_s \frac{\left(x_{i,s} - y_{i,s}\right)^2}{y_{i,s}} + \sum_s \left(x_{i,s} - y_{i,s}\right),
\end{eqnarray*}
which using the fact that $\sum_s x_{i,s} = \frac{1}{n} = \sum_s y_{i,s}$ becomes
$$\sum_s \frac{\left(x_{i,s} - y_{i,s}\right)^2}{y_{i,s}} \le \frac{n}{\Lambda} \|x_i - y_i\|_2^2,$$
as we have $y_{i,s} \ge \frac{\Lambda}{n}$ for any $y_i \in \bK_i$. This shows that $\frac{\Lambda}{n} \BRi{x_i,y_i} \le \|x_i - y_i\|_2^2$ for any $x_i, y_i \in \bK_i$. On the other hand, for any $x_i, y_i \in \mc{K}_i$,
$$\|x_i - y_i\|_2^2 \le \|x_i - y_i\|_1^2 \le 2 \BRi{x_i,y_i},$$
by Pinsker's inequality. Thus, we have the proposition.
\end{document}